\documentclass[amsmath,amsfonts,amssymb]{sig-alternate}

\begin{document}

\title{Resolution And Datalog Rewriting \\ 
Under Value Invention And Equality Constraints
\\ 
}

\author{
Bruno Marnette\\
{\large $\mathsf{marnette.bruno}@\mathsf{gmail.com}$}
\\ \ \\
\normalsize{{\bf Disclaimer}: This preliminary report was originally written at the INRIA Saclay in 2010.}  \\
\normalsize{The list of references may therefore be incomplete at the time of this release (Dec 2012)}}

\maketitle

\newtheorem{lemma}{Lemma}
\newtheorem{definition}{Definition}
\newtheorem{theorem}{Theorem}
\newtheorem{observation}{Observation}
\newtheorem{proposition}{Proposition}

\newtheorem{fact}{Fact}

\newtheorem{corollary}{Corollary}
\newtheorem{example}{Example}

\newcommand\Var{\mathcal{V}}
\newcommand\var{\mathcal{V}}
\newcommand\entails{\models}
\newcommand\eqdef{\stackrel{_{\mathsf{def}}}{=}}
\newcommand\chase{\stackrel{_{\!\!\!\mathsf{chase}}}{\longrightarrow}\!\!} 
\renewcommand\flat[1]{\stackrel{_{ #1 }}{\longrightarrow}\!\!} 
\newcommand\parachase{\stackrel{_{\mathsf{chase}}}{\longrightarrow}\!\!} 

\newcommand\res{\stackrel{_{\!\!\mathsf{resol}}}{\longrightarrow}\!\!} 
\newcommand\parares{\stackrel{_{\mathsf{Resol}}}{\longrightarrow}\!\!}

\newcommand\starrow{\stackrel{_{\diamond}}{\longrightarrow}\!\!} 
\newcommand\parastarrow{\stackrel{ _{\diamond}} {\longrightarrow} \!\!} 

\newcommand\Fix{\mathsf{Fix}}

\newcommand\I{\mathcal I}
\newcommand\J{\mathcal J}
\newcommand\Q{\mathcal Q}
\newcommand\R{\mathcal R}
\newcommand\D{\mathcal D}
\newcommand\U{\mathcal U}
\newcommand\G{\mathcal G}

\newcommand\A{\sigma_{\!A}}

\newcommand\resmatch[2]{\mathsf{Match}^\mathsf{rslv}(#1,#2)} 
\newcommand\chasematch[2]{\mathsf{Match}^\mathsf{chs}(#1,#2)}
\newcommand\somematch[2]{\mathsf{Match}^\diamond(#1,#2)}
 
\newcommand\Chase{\mathsf{Chase}}
\newcommand\Resol{\mathsf{Resol}}
\newcommand\depth{\mathsf{Depth}}

\newcommand\Flat{\mathsf{Flat}}

\newcommand\chaseto[1]{\stackrel{{}_{#1}}{\longrightarrow}}

\newcommand\Core{\mathsf{Core}}

\newcommand\Ans{\mathsf{Ans}}

\newcommand\FlatChase{\mathsf{Ans}}

\newcommand\FlatRes{\mathsf{FRes}}

\newcommand\lw{\mathsf{lw}}
\newcommand\gw{\mathsf{gw}}

\newcommand\ang[1]{\langle #1 \rangle}

\newcommand\entailsinf{\,{\models^{\!\!\!\!\!^\infty}}\,}

\newcommand\X{{^{_X}}\!}
\newcommand\Y{{^{_Y}}\!}
\newcommand\Z{{^{_Z}}\!}

\renewcommand\rho{\mathit{pos}}
\newcommand\pos{\mathit{pos}}

\newcommand\SatRes{\mathsf{SRes}}

\newcommand\UCQ{\mathsf{UCQ}}

\newcommand\Certain{\mathsf{Certain}}

\newcommand\M{\mathcal M}
\newcommand\Source{\mathbf S}
\newcommand\Target{\mathbf T}

\newcommand{\V}[1]{{{\mathcal V}_{#1}}}
\newcommand\ok{\mathsf{ok}}

\newcommand\Dom{\mathsf{Dom}}
\newcommand\tw{\mathsf{tw}}

\newcommand\GAV{\mathsf{Gav}}

\newcommand\Ref{\mathsf{Ref}}
\newcommand\CRef{\mathsf{CRef}}

\begin{abstract}
While Datalog is a golden standard for denotational query answering, it does not support  value invention or equality constraints. The \emph{Datalog}$^\pm$ framework introduced by Gottlob faces these issues by considering rules with fresh variables in the head (known as \emph{tgds}) or equalities in the head (known as \emph{egds}). Several tractable classes have been identified, among which: (S) the class of \emph{sticky} tgds; (T) the class of tgds and egds ensuring oblivious termination; and (G) and the class of \emph{guarded} tgds. In turn, the tractability of these classes typically relies on the `chase': (S) ensures that every chase derivation is `sticky'; (T) ensures polynomial chase termination;  and (G) allows stopping the chase after a 
fixed `depth' while preserving completeness. 

This paper shows that there are alternative algorithms (instead of the chase) that can serve as a basis for the design of (larger) tractable classes. As a first contribution, we present an algorithm for \emph{resolution} which is complete for any set of tgds and egds (rather than being complete only for specific subclasses). We then show that a technique of saturation can be used to achieve completeness with respect to First-Order (FO) query rewriting. 
As an application, we generalize a few existing classes (including (S)) that ensure the existence of a finite FO-rewriting. 

We then consider a more general notion of rewriting, called \emph{Datalog rewriting}, and show that it provides a truly unifying paradigm of tractability for the family of Datalog$^\pm$ languages. While the classes (S), (T) and (G) are incomparable, we show that every set of rules in (S), (T) or (G) can be rewritten into an equivalent set of standard Datalog rules. On the negative side, this means intuitively that Datalog$^\pm$ does \emph{not} extend the expressive power of Datalog in the context of query answering. On the positive side however, one may use the flexible syntax of Datalog$^\pm$ while using (only) standard Datalog in the background, thus making use of existing optimization techniques, such as Magic-Set. 
\end{abstract}

\section{Introduction}

While the language Datalog and its extensions have been studied for decades (see e.g. \cite{CeGT90,AV91,AbHV95}), they  recently received a renewed attention. In particular Gottlob et al introduced a comprehensive and unifying framework called Datalog$^\pm$ (\cite{CGLMP10,CaGK08,CaGL09,CGP10}) which is based on a family of Datalog extensions. One of the main qualities of this framework lies in its generality. In particular, Datalog$^\pm$ is expressive enough to cover some interesting classes of ontologies, some light-weight description logics  and some fragments of F-logic (see e.g. \cite{CGLMP10}). It was also argued in \cite{CGLMP10} that Datalog$^\pm$ is useful in a variety of contexts, for example Data Exchange \cite{FKMP05} and the Semantic Web \cite{SemWeb}. The main reason for this, is that Datalog$^\pm$, unlike standard Datalog, addresses the fundamental problem of \emph{value invention} by considering rules with fresh variables in the head. Such rules are known as \emph{tuple-generating dependencies} (\emph{tgds} in short) and correspond to first-order formulas of the form 
$$
\forall \bar x,\bar y,\ \phi(\bar x,\bar y) \to \exists \bar z, \psi(\bar x,\bar z)
$$  
where $\phi$ and $\psi$ are two conjunctions of atoms (that may express \emph{joins}) while $\bar z$ is a tuple of existentially quantified variables that can be used to reason about unknown entities. In addition, Datalog$^\pm$ supports \emph{equality-generating dependencies} (\emph{egds} in short) of the form 
$$
\forall x,y,\bar z,\  \phi(x,y,\bar z) \to x=y
$$  
where $\phi$  is a conjunction of atoms. Tgds and egds can be used to encode typical schema dependencies such as \emph{inclusion dependencies} or \emph{function dependencies} (see e.g. \cite{AbHV95}) which, in turn, allow to reason about structured data. On the negative side, however, the problem of query answering under arbitrary tgds and egds is undecidable \cite{BeVa81}, and it was observed in \cite{CaGK08} that the problem remains undecidable even for a \emph{fixed} set tgds. 

To avoid undecidability, the Datalog$^\pm$ framework typically considers (in \cite{CGLMP10}) three alternative restrictions called \emph{termination}, \emph{guardedness} and \emph{stickiness}, defined as follows:\\

\noindent{\bf Termination:} There exists a chase procedure  \cite{FKMP05,ChaseRevisited,CaGK08,brunoPODS} that, for a given database, computes a universal solution in polynomial time (data complexity) which, in turn, can be used for sound and complete query answering.  \\

\noindent {\bf Guardedness}: The set of dependencies consists of \emph{guarded} tgds and \emph{separable} egds, in which case, as shown in \cite{CaGK08}, it is possible to reach tractability, while remaining complete, by stopping the oblivious chase after a fixed depth.  \\

\noindent {\bf Stickiness}: The set of dependencies is a set of tgds (and separable egds) ensuring that every chase derivation is \emph{sticky} \cite{CGP10}, meaning intuitively that the fresh variables introduced by the chase are only propagated in an harmless way.   
\\
\newpage

The three classes discussed above are unfortunately incomparable, and the only property that really unifies them in \cite{CGLMP10} is the fact that they are all based on the \emph{chase}, either in their definition or in terms of properties. This approach has several advantages and it contributes, in particular, to the simplification of the `big picture'. However, there are alternative ways of approaching the above classes. 
 For example, the guarded fragment can be decided, at least in some cases, using \emph{tableaux} algorithms or \emph{resolution} algorithms (see e.g. \cite{Nivelle03,KazNiv04, H02,GuardProc}). Similarly, the criterion of stickiness can be understood as a criterion ensuring the termination of resolution (as opposed to a specific property of the chase). In fact, a custom resolution procedure was proposed in \cite{CGP10} for the case of stickiness. However, this procedure relies on the specific properties of sticky tgds, and it is relevant only for this class of dependencies. In contrast with \cite{CGP10}, a contribution of this paper will be to consider a resolution algorithm which is defined (and complete) for arbitrary sets of tgds and egds, as to identify larger tractable classes.\\

While exploring some alternatives to the chase procedure, this paper follows a similar methodology to that of Datalog$^\pm$ in the sense that it aims at identifying a \emph{unifying} paradigm. Towards this goal, we will consider, beside resolution, several notions of \emph{rewriting}: 

\begin{itemize}
\item {\bf Data Rewriting} \emph{(a.k.a. Universal Solutions)}. \\ 
Given a database $\D$ and a set of dependencies $\Sigma$, a \emph{data rewriting} for $\D$ is a new database $\D_\Sigma$ which integrates all the information can be inferred from $\Sigma$. Given such a rewriting $\D_\Sigma$, we can then test whether a conjunctive query $\Q$ is implied by $\D$ and $\Sigma$ by testing whether $\Q$ is implied by $\D_\Sigma$. 
This technique can be used whenever the chase terminates since a \emph{universal solution} is in fact a special case of data rewriting. 
\item {\bf Query Rewriting} \emph{(a.k.a. First-Order Rewriting)}.\\
Given a query $\Q$ and a set of dependencies $\Sigma$, a \emph{query rewriting} for $(\Sigma,\Q)$ is a query $\Q_\Sigma$ (in a  first-order language) which integrates all the information  encoded in  $\Sigma$. Given such a rewriting $\Q_\Sigma$, we can test whether $\Q$ is implied by $\Sigma$ and a database $\D$ by testing whether $\Q_\Sigma$ is implied by $\D$. As already observed in \cite{CGP10}, the technique of query rewriting can be used whenever $\Sigma$ is a set of sticky tgds. 
\item {\bf Datalog Rewriting} \emph{(The Unifying Paradigm)}.\\ 
Given a conjunctive query $\Q$ and a set of dependencies $\Sigma$, a Datalog rewriting for $(\Sigma,\Q)$ consists of a new pair $(\Sigma',\Q')$ where $\Sigma'$ is a set of standard Datalog rules (without fresh variables in the head) and $\Q'$ is a query such that $(\Sigma',\Q')$ is equivalent to $(\Sigma,\Q)$ with respect to query answering. 
\end{itemize}
 
As we will show in this paper, the technique of Datalog rewriting is not only a strictly generalisation of (first-order) query rewriting, but it can also be used in the case of terminating dependencies (instead of data rewriting) and in the case of guarded dependencies (instead of relying on the chase). In this sense, Datalog rewriting is therefore a truly unifying paradigm for the family of Datalog$^\pm$ languages. 
Also, Datalog rewriting was proved useful for classes of dependencies that are not (yet) covered by the Datalog$^\pm$ framework, in particular in the context of Description Logic (see e.g. \cite{pumh08rewriting}). In this sense, heuristics for Datalog rewriting can be used to further generalize the Datalog$^\pm$ framework. 

\subsection*{Structure and Main Contributions}

  The preliminary section formalizes the problem of query answering (under constraints) as a basic implication problem. For the sake of concision and symmetry, the key notions of \emph{databases}, \emph{queries} and \emph{dependencies} are all based of sets of atoms with variables only (called \emph{instances}). Constants, null values, and non-boolean queries with equality atoms are discussed in Section 3.5.   \\

\noindent {\bf (1)} In Section 3.1, we propose a concise definition of \emph{resolution} for arbitrary tgds and egds and show that it is complete for query answering. In contrast with alternative approaches (such as \cite{CompSkol}), this definition does \emph{not} rely on skolemization and unskolemization. Instead, it relies simply on instances and renamings  (a.k.a. homomorphisms). \\

\noindent {\bf (2)} In Section 3.2, we show that a technique of saturation can be used to reach completeness with respect to finite rewritability. More precisely, the saturated chase (like the core chase of \cite{ChaseRevisited}) computes a finite data rewriting whenever there exists one. Similarly, the saturated resolution computes a finite query rewriting whenever there is one. \\

\noindent {\bf (3)} In Section 3.3, we focus on tgds and revisit the notion of stickiness. We observe that there is a reduction from the class of sticky tgds to the class of lossless tgds which, in fact, can be used to identify larger classes or rewritable queries. \\

\noindent {\bf (4)} In Section 3.4, we show that query rewriting under egds (as opposed to tgds only) is more complex. We then propose a heuristic for the integration of a set of egds into a set of tgds, thus allowing (in some cases) to rely on stickiness and query rewriting despite the presence of conflicting egds. \\

\noindent {\bf (5)} In Section 4.1, we show that data rewriting captures the class of dependencies ensuring the termination of the oblivious chase (from \cite{brunoPODS}), and as a special case, the class of weakly acyclic dependencies (from \cite{FKMP05}). \\

\noindent {\bf (6)} In Section 4.2, we finally show that data rewriting also captures the class of guarded tgds. This result is the most technical and arguably the most important contribution.\\

The proofs (or more accurately, the proof sketches) have been included to the body of the paper. Readers unfamiliar with the problem of query answering under tgds and egds are invited to consult e.g. \cite{CGLMP10} for more examples, applications and related works. Additional proof details (for some of the proofs) can be found in \cite{mythesis}.

\section{Preliminaries}

We chose the \emph{infinite model} semantic. Recall however that it coincides with the \emph{finite model} semantic under either termination (\cite{ChaseRevisited,brunoPODS}),  guardedness (\cite{BaGO10}), or  stickiness (\cite{CGP10}). \\

\noindent {\bf Instances and Dependencies.} Let $\Var$ be a countable set of variables and let $\sigma$ be a finite set of predicates. We assume that each $R \in \sigma$ comes with a fixed and finite arity $a_R$.  
An instance $I$ is a set of atoms $R(\bar v)$ where $R \in \sigma$ and $\bar v$ is a tuple of variables respecting the arity of $R$, that is, $\bar v \subseteq \Var^{a_R}$. We let $\V{I}$ be the set of variables occurring in an instance $I$. A \emph{tgd} is a rule $B \to H$ where $B$ and $H$ are two finite instances called the \emph{body} and the \emph{head} of $r$, respectively. We say that a tgd $r$ is \emph{of the form} 
$
 B(X,Y) \to H(X,Z)
$
when $X = \V B \cap \V H$,  $Y= \V B {\setminus} X$, and $Z = \V H {\setminus} X$. Such a tgd $r$ is called a \emph{Datalog rule} when $Z=\emptyset$. An \emph{egd} $r$ is a rule $B \to x=y$ where $B$ is a finite instance and $x,y \in \V B$. An egd $r$ is \emph{of the form} 
$B(x,y,Z) \to x=y$ when $Z = \V B{\setminus}\{x,y\}$.\\

\noindent {\bf Semantics.} 
A \emph{renaming} is a mapping from $\Var$ to $\Var$. Given an instance $I$ and a renaming $\theta$, we let $I[\theta]$ be the set of atoms $R(\theta(t_1),\ldots,\theta(t_n))$ where $R(t_1,\ldots,t_n) \in I$. 
Given two instances $I$ and $J$ we let $I \entails J$ when $J[\theta] \subseteq I$ for some renaming $\theta$. 
Given an instance $I$ and a set $\J$ of instances, we let $I \entails \J$ when $I \entails J$ for some $J \in \J$. Given two sets $\I$ and $\J$ of instances, we let $\I \entails \J$ when $I \entails \J$ for all $I \in \I$. Note that these definitions are consistent with respect to singletons. In particular, we have $I \entails J$ iff $\{I\} \entails \{J\}$. Given two (sets of) instances $\I$ and $\J$ we write $\I \equiv \J$ when both $\I \entails \J$ and $\J \entails \I$.

Given $X \subseteq \Var$, a renaming \emph{of $X$} is a renaming $\theta$ such that $\theta(v)=v$ for all $v \in (\Var\setminus X)$. In the following, we use the notation $\theta_X$ to indicate that $\theta_X$ is a renaming of $X$. 
 Given $Y\subseteq \Var$ disjoint from $X$, and given a renaming $\theta_Y$, we denote by $[\theta_X,\theta_Y]$ the  renaming of $X \cup Y$ that coincides with $\theta_X$ on $X$ and with $\theta_Y$ on $Y$.  

Given an instance $I$ and a tgd $r:B(X,Y) \to  H(X,Z)$, we let $I \entails r$ when for all $\theta_{X}$ and $\theta_Y$  such that $B[\theta_{X},\theta_Y] \subseteq I$, there exists $\theta_Z$ such that $H[\theta_X,\theta_Z] \subseteq I$. Given an instance $I$ and an egd $r:B(x,y,Z) \to x=y$, we write $I \entails r$ when for all $\theta_{x}$, $\theta_y$ and $\theta_Z$ such that $B[\theta_x,\theta_y,\theta_Z] \subseteq I$, it holds that $\theta_{x}(x)=\theta_{y}(y)$. Given a set $\Sigma$ of tgds and egds, we finally write $I \entails \Sigma$ when $I \entails r$ for all tgds and egds $r \in \Sigma$. \\

Note that, for two instances $I$ and $J$, the property $I \entails J$ corresponds to several equivalent intuitions:
 the \emph{boolean conjunctive query} $J$ is true in the \emph{database} $I$; the query $I$ is \emph{contained} in the query $J$; the \emph{formula} $J$ is \emph{implied} by $I$; the instance $I$ is a \emph{model} of $J$; or there is an \emph{homomorphism} from $J$ to $I$. A similar comment holds for sets  of instances.  

In the following definition, $\D$ and $\Q$ denote two sets of instances and $\Sigma$ denotes a set of dependencies. Intuitively: $\D$ corresponds to a database or a set of databases (representing a set of possible models $M$), 
$\Sigma$ corresponds to an ontology or a set of structural dependencies; and $\Q$ corresponds to a union of boolean conjunctive queries.  

\begin{definition}
We say that $\Q$ is certain in $\D$ under $\Sigma$, denoted $D \land \Sigma \entails \Q$, iff, for all instance $M$ such that $M \entails \D$ and $M \entails \Sigma$, we have $M \entails \Q$.
\end{definition}

\noindent {\bf Chase.} 
We next recall the definition of the \emph{chase} \cite{FKMP05,ChaseRevisited,CaGK08,brunoPODS}. More precisely, the following definition  coincides with \cite{CaGK08} in the case of tgds and with \cite{FKMP05} in the case of egds.

Given an instance $I$ and a set of variables $Z$, we say that $\theta_Z$ is an $I$-fresh renaming when 
 $\theta_Z$ is a renaming of $Z$, $\theta_Z$ is injective on $Z$ and its image $\theta_Z(Z)$ is disjoint from $\V{I}$. 
Given two variables $u$ and $v$ we denote by $[u{\gets}v]$ the unique renaming $\theta$ of $\{v\}$ satisfying 
 $\theta(v) = u$. Given an instance $I$, 
 and a  set $\Sigma$ of tgds and egds, we  write $I \chase_\Sigma\ J$ when either $J=I$ or one of the following rules applies:
\pagebreak 
\begin{itemize}
\item[{\small(tgd)}] There is a tgd $r:  B(X,Y) \to H(X,Z)$ in $\Sigma$ 
 and some renamings $\theta_{X}$ and $\theta_Y$ such that $B[\theta_X,\theta_Y] \subseteq I$ and $J = I \cup H[\theta_X,\theta_Z]$ for some $I$-fresh renaming $\theta_Z$. 
\item[{\small(egd)}] There is an egd $r:B(x,y,Z) \to x=y$ in $\Sigma$ and some renamings $\theta_x$, $\theta_y$, and $\theta_Z$ such that $B[\theta_x,\theta_y,\theta_Z] \subseteq I$ and 
 $J = I[\theta_y(y) \gets \theta_x(x)]$.
\end{itemize} 
Using the symbol $*$ for the transitive closure, we finally write $J \in \Chase(I,\Sigma)$ when $I \chase_\Sigma^*\ J$. 

\section{First-Order Resolution}

\subsection{Definition and Completeness}

We next propose a  definition of \emph{resolution} for tgds and egds which, unlike alternative approaches (such as \cite{CompSkol}), does not require any complex algorithm of (un)skolemization.

 Given an instance $Q$ and a set $\Sigma$ of tgds and egds, we 
write $Q \res_\Sigma\, R$ when one of the following rules applies:
\begin{itemize}
\item[{\small(ren)}] $R = Q[\theta]$ for some renaming $\theta$. 
\item[{\small(tgd)}] There is a tgd $B(X,Y) \to H(X,Z)$ in $\Sigma$ and three renamings $\theta_X$, $\theta_Y$, $\theta_Z$ such that 
$$
R = (Q \setminus H[\theta_X,\theta_Z]) \cup B[\theta_X,\theta_Y] 
$$
and the following conditions hold:\\
- there is at least one atom in $Q \cap H[\theta_X,\theta_Y]$;\\ 
- $\theta_Z$ is an injection from $Z$ to $(\Var \setminus \Var(R))$.
\item[{\small(egd)}] There is an egd $B(x,y,Z) \to x=y$ in $\Sigma$ and three substitutions $\theta_x$, $\theta_y$, $\theta_Z$, such that 
$\theta_x(x)$ occurs several times in $B$, $\theta_x(x) \not= \theta_y(y)$, and $R$ is obtained by:\\
- first renaming some occurrences of $\theta_x(x)$ into $\theta_y(y)$;\\ 
- and then adding the atoms of $B[\theta_x,\theta_y,\theta_Z]$. 
\end{itemize}
We finally write $R \in \Resol(Q,\Sigma)$ when $Q \res_\Sigma^*\ R$.

\begin{theorem}
For all set $\Sigma$ of tgds and egds, and for all instances $D$ and $Q$, the following statements are equivalent: 
\begin{itemize}
\item[(1)] $D \land \Sigma \entails Q$ 
\item[(2)] $\exists\, U \in \Chase(I,\Sigma),\ U \entails Q$
\item[(3)] $\exists\, R \in \Resol(Q,\Sigma),\ D \entails R$ 
\end{itemize}
\end{theorem}

\begin{proof}[Sketch] The equivalence between  (1) and (2) is well-known. Recall however that it holds only under the infinite model semantics (see e.g. \cite{ChaseRevisited} for more details).  
The implication (3) $\Rightarrow$ (1) means that the resolution is \emph{sound}, and it is fairly straightforward. We next prove that (2) implies (3). We proceed by induction, and show that the following property holds for all $n \ge 0$:
\begin{itemize}
\item[] If there is a chase derivation of length $n$ of the form 
$$D=I_0 \chase_{\Sigma}\ I_1 \chase_{\Sigma}\ \ldots \chase_\Sigma\ I_n \ \mbox{where} \ I_n \entails Q$$
then there exists $R \in \Resol(Q,\Sigma)$ such that $D \entails R$. 
\end{itemize}
In the case  $n=0$, the property holds with $R=Q$. Assume now the property true for $n-1$ and consider a derivation of length $n$ as above. Let $h$ be a renaming such that $Q[h] \subseteq I_n$ and let $Q' = Q[h]$. Because of the resolution rule 
(ren), we have $Q \res_\Sigma\ Q'$. We now distinguish two cases: \\

{\bf (1)} Assume that $I_n$ is obtained by chasing a tgd $$B(X,Y) \to H(X,Z)$$ and consider $\theta_X$,$\theta_Y$,$\theta_Z$  such that $B[\theta_X,\theta_Y] \subseteq I_{n-1}$ and $I_n = I_{n-1} \cup H[\theta_X,\theta_Z]$. If $Q' \subseteq I_{n-1}$, we can apply the induction hypothesis for $n-1$. Otherwise, there is at least one atom 
in $Q' \cap H[\theta_X,\theta_Z]$ and we have $Q' \res_\Sigma\ R$ for 
$$
R = (Q' \setminus H[\theta_X,\theta_Z]) \cup B[\theta_X,\theta'_Y].  
$$
Since $R \subseteq I_{n-1}$ we have $I_{n-1} \entails R$ and we can easily conclude the inductive step. \\

{\bf (2)} Assume that $I_n$ is obtained by chasing an egd
$$B(x,y,Z) \to x=y$$ 
and consider $(\theta_x,\theta_y,\theta_Z)$ such that $B[\theta_x,\theta_y,\theta_Z]\subseteq I_{n-1}$ and 
 $I_{n} = I_{n-1}[\theta']$ for $\theta'=[\theta_y(y) {\gets} \theta_x(x)]$. For each atom $A \in Q$, choose an atom 
$c(A) \in I_{n-1}$ such that $\theta'(c(A))=A$ and let $R = \{c(A), A \in Q'\} \cup B[\theta_x,\theta_y,\theta_Z]$. 
We can finally check that $Q' \res_\Sigma\ R$ and conclude the inductive step. 
\end{proof}

\subsection{Rewriting and Saturation}

We next formalize the notions of data rewriting and query rewriting before showing that completeness (with respect to finite rewritability) can be reached, in both cases, by \emph{saturating} the relation $\chase\ $ or $\res\ $ in a rather natural way.

\begin{definition}
Consider a set $\Sigma$ of dependencies and two sets $\D$ and $\Q$ of instances. 
A \emph{data rewriting} for $(\D,\Sigma)$ is a set of instances $\U$ such that, for all set $\Q'$ of instances:
$$
(\D \land \Sigma \entails \Q')  \ \ \mbox{iff} \ \ (\U \entails \Q'). 
$$  
Similarly, a \emph{query rewriting} for $(\Sigma,\Q)$ is a set of instances $\R$ such that, for all set $\D'$ of instances:
$$
(\D' \land \Sigma \entails \Q)  \ \ \mbox{iff} \ \ (\D' \entails \R). 
$$  
\end{definition}

{\bf Saturation.} Let $\starrow\ \,\in \{\chase\ ,\res \ \}$. Given two finite sets of instances $\I$ and $\J$, we write $\I \starrow \J$ when, intuitively, $\J$ is a concise representation (up to equivalence) of all the instances $J$ such that  $I \starrow J$ for some $I \in \mathcal I$. More formally, we let $\I \starrow \J$ when $\J$ can be returned by the following non-deterministic algorithm: 
\begin{center}
\begin{tabular}{l}
Start with $\J := \I$. \\
Repeat (until fixed point)\\  
\quad  If there is an instance $J$ such that: \\
\qquad (1) $I \starrow_\Sigma\ J$ for some $I \in \I$; and \\
\qquad (2.1) we are in the case $\starrow\ \ =\  \chase\ $ and\\ 
\qquad\qquad\  there is no $J' \in \J$ such that $J \entails J'$; or\\ 
\qquad (2.2) we are in the case $\starrow\ \ =\  \res\ $ and \\
\qquad\qquad\  there is no $J' \in \J$ such that $J' \entails J$\\
\quad Then let $\J:= \J \cup \{J\}$\\
Return $\J$.
\end{tabular} 
\end{center}

Note that the above definition can easily by translated into an actual algorithm since the set of instances $I$ satisfying the point (1) is finite up to isomorphism (and can be enumerated in exponential time). We define a \emph{complete derivation} for $\I$ and $\starrow$\ \  as an infinite series $(\I_i)_{i \ge 0}$ where $\I_0 = \I$ and $\I_i \starrow \I_{i+1}$ for each $i\ge 0$. We finally let $\J \in \Fix(\I,\starrow\ )$ when 
 $\J = \bigcup_i \I_i$ for some complete derivation $(\I_i)_{i \ge 0}$.

\begin{theorem} \ 
For all set $\Sigma$ of tgds and egds, and for all sets $\D$ and $\Q$ of instances: 
\begin{itemize}
\item Every $\U \in \Fix(\D,\chase_\Sigma )$ is a data rewriting of $(\D,\Sigma)$. 
\item Every $\R \in \Fix(\Q,\res_\Sigma )$ is a query rewriting of $(\Sigma,\Q)$.
\end{itemize}
\end{theorem}

\begin{proof}[Sketch]
The result would follow directly from Theorem 1 if we had removed the requirements (2.1) and (2.2) in the above definition of saturation. To prove that the result still holds with (2.1) and (2.2), it is enough to observe that the chase and the resolution are both monotone: 
\begin{itemize}
\item If $J \entails J'$\  and $J'\chase_\Sigma\  K'$, then 
there exists an instance $K$ such that $J \chase^*_\Sigma\ K$ and $K \entails K'$.
\item If $J' \entails J$ and $J' \res_\Sigma\ K'$, then there exists
an instance $K$ such that $J \res^*_\Sigma\ K$ and $K' \entails K$.
\end{itemize}
The monotonicity of the chase is well-known (see e.g. \cite{brunoPODS}). For the monotonicity of the resolution, consider $J' \entails J$ and $J' \res_\Sigma\ K'$. The case where $K'$ is obtained from $J'$ with the resolution rule (ren) is straightforward. 
Consider now the case (tgd) and unfold the definition of resolution so that  
$$     
K' = (J' \setminus H[\theta_X,\theta_Z]) \cup B[\theta_X,\theta_Y]. 
$$
Let $h$ be a renaming of $\V J$ such that $J[h] \subseteq J'$. If $J[h]$ does not intersect $H[\theta_X,\theta_Z]$ 
 we have $J[h] \subseteq K'$ and the property holds for $K = K'$. Otherwise, considering the instance  
$$     
K = ((J[h]) \setminus H[\theta_X,\theta_Z]) \cup B[\theta_X,\theta_Y]. 
$$
we can check that $I \res_\Sigma\ I[h] \res_\Sigma\ K$ and $K[h] \subseteq K'$. Consider finally the case (egd) and unfold the definition of resolution for an egd $B(x,y,Z) \to x=y$ of $\Sigma$ and three substitutions $\theta_x$, $\theta_y$, $\theta_Z$. For each atom $A \in J'$, let $c(A)$ be the atom of $K'$ that has been obtained from $A$ by renaming some occurrences of $\theta_x(x)$ into $\theta_y(y)$, and observe that: 
$$K' = \cup \{c(A) | A \in J'\} \cup B[\theta_x,\theta_y,\theta_Z].$$
Let $h$ be a renaming of $\V J$ such that $J[h] \subseteq J'$ and let 
$$
K = \cup \{c(A) | A \in J[h]\} \cup B[\theta_x,\theta_y,\theta_Z]. 
$$
As in the previous case, we can finally check that we have $I \res_\Sigma\ I[h] \res_\Sigma\ K$ and $K[h] \subseteq K'$. This concludes the proof of monotonicity and the proof of Theorem 2.
\end{proof}

{\bf Termination.}  Let $(\I_i)_{i \ge 0}$ be a complete derivation of $\I$ and $\starrow\ $. We say that this derivation \emph{terminates} iff there exists a finite $i$ such that $I_{i+1} = I_i$. Note that, in this case, we also have $I_j = I_i$ for all $j \ge i$. We finally say that $\Fix(\I,\starrow_\Sigma)$ is finite when \emph{all} the instances $J \in \Fix(\I,\starrow_\Sigma)$ are finite, meaning (equivalently) that \emph{all} the derivations of $\I$ by $\starrow_\Sigma$ are finite.

\begin{theorem} \ 
Given a finite set $\Sigma$ of tgds and two finite sets $\D$ and $\Q$ of instances, the following statements hold: 
\begin{itemize}
\item If there exists a finite data rewriting for $(\D,\Sigma)$,\\
 then $\Fix(\D,\chase_\Sigma\ )$ is finite.
\item If there exists a finite query rewriting for $(\Sigma,\Q)$,\\
then $\Fix(\Q,\res_\Sigma\ )$ is finite. 
\end{itemize}
\end{theorem}

\begin{proof}[Sketch]
The two points are similar and we prove here the second one. 
Suppose that there exists a finite query rewriting $\R$ for $(\Sigma,\Q)$ and a consider a complete derivation 
 $(\I_i)_{i \ge 0}$ for $\Q$ and $\res_\Sigma$. Since $\R$ and $\cup i \I_i$ are two rewritings of $(\Sigma,\Q)$, we have 
$\R \equiv \cup i \I_i$. Since $\R$ is finite, there exists a finite $k$ such that $\R \equiv \I_k$. We can then observe that $\I_k = \I_{k+1}$ and conclude that $(\I_i)_{i \ge 0}$ is finite.
\end{proof}

 Note that a direct consequence of Theorems 2 and 3 is the following: as soon as there is \emph{one} finite fixed point in $\Fix(\D,\chase_\Sigma)$ or $\Fix(\Q,\res_\Sigma)$, it is the case that \emph{all} the fixed points are finite. In particular, this means that all saturation strategies are equivalent with respect to termination, both in the case of $\chase\ $ and $\res\ $.

\subsection{Rewritable Classes of Tgds}

This section  illustrates how the resolution procedure from Section 3.1 can be used to `simplify' the big picture on query rewritability. In particular, we first show that it provides a concise proof of finite query-rewritability of a few well-known classes of dependencies. In fact, the results stated in the following proposition are rather well-know. It is however interesting to compare the following proof sketch (based on resolution, and arguably simple) with, for example, the seminal paper of Johnson and Klug \cite{JoKl84} where a proof was given (based on the chase, and arguably complex) for the tractability of \emph{inclusion dependencies}. This class of dependencies is indeed covered (strictly) by the class of \emph{local-as-view} tgds (\emph{lav} tgds) defined in the following proposition. 

\begin{proposition} \label{prop:termres}
In each of the following cases, we can effectively compute a finite query rewriting for $(\Sigma,\Q)$:
\begin{itemize}
\item \emph{(Lav tgds)}  $\Sigma$ is a set of tgds $B \to H$ where the body $B$ contains at most one atom.  
\item \emph{(Lossless tgds)} $\Sigma$ is a set of tgds $B \to H$ where, for each atom $A_h \in H$, we have $\V B \subseteq \V{A_h}$.  
\item \emph{(Acyclicity)} $\Sigma$ is a set of tgds and there is a linear order $\le_\Sigma$ on the predicates of $\Sigma$ such that, 
for all tgd $B \to H$ in $\Sigma$, all predicate $R_b$ occurring in $B$, and all predicate $R_h$ occurring in $H$, we have $R_h \le R_b$.  
\end{itemize}
\end{proposition}

\begin{proof}[Sketch] In the case of lav tgds, we can observe that the resolution rules ({hom}) and ({tgd}) never increase the number of atoms. 
More formally, whenever $I \res_\Sigma\ J$, we have $|J| = |I|$. There are therefore a finite number of instances (up to $\equiv$) that can be computed by resolution. As a consequence, every (complete) derivation terminates and every $\R \in \Fix(\D,\res_\Sigma)$ is a finite query rewriting. 
In the case of lossless tgds, the result follows from a very similar observation: the resolution rules ({hom}) and ({tgd}) never increase the number of variables. 
(That is, $I \res_\Sigma\ J$ implies $\V J \subseteq \V I$). Consider finally the case of an acyclic set $\Sigma$ of tgds. Let $\sigma_\Sigma$  be the set of predicates occurring in $\Sigma$ and consider the ordering $\sigma_\Sigma = \{R_1, \ldots,R_n\}$ where $R_{i-1} \le_\Sigma R_{i}$ for each $i \le n$.
For each instance $I$, consider $s(I) = (a_1,\ldots,a_{n})$ where, for each $i \in \{1,..,n\}$, $a_i$ is the number of atoms $R'(\bar v) \in I$ where $R'=R_i$. 
We can observe that, whenever $I \res_\Sigma\ J$, the tuple $s(J)$ is smaller than $s(I)$ with respect to lexicographic order. 
We can finally conclude as in the previous cases. 
\end{proof}

We next shed more light on the notion of stickiness from \cite{CGP10} which was discussed in the introduction. 
In particular, we show that there is in fact a direct reduction (preserving the property of finite rewritability) from the class of sticky tgds to the class of lossless tgds. This reduction proves useful in two ways: (1) it provides a direct proof of rewritability for the class of sticky tgds (which, unlike \cite{CGP10}, does not require a custom resolution algorithm), (2) and it also allows us to identify a more general class of rewritable dependencies. 
 In a nutshell, the key idea of the following reduction consists in replacing an atom $A(\bar x,\bar y)$ by an atomic formula $R(\bar x) \approx (\exists \bar y, A(\bar x,\bar y))$ whenever $A(\bar x,\bar y)$ is the only atom (in a given tgd) where the variables of $\bar y$ occur.  \\
 
\noindent {\bf Simplifying atoms}. Given a tgd $r:B {\to} H$ and an atom $A$ in the body $B$, we let $X_{A,r} = \V A \cap \V H$ and $Y_{A,r} = \V A {\setminus} X_{A,r}$. We then say that $A$ 
\emph{can be simplified} in $r$ when $Y_{A,r}$ is non-empty and disjoint from $\V {B \setminus A}$. For example, in the tgd
$$
A(x_1,y_1), B(x_1,x_2,y_2), C(x_1,y_3),D(y_2) \to R(x_1,x_2,z_1)
$$ 
the atoms $A(x_1,y_1)$ and $C(x_1,y_3)$ can be simplified. In contrast, $B(x_1,x_2,y_2)$ cannot be simplified because $y_2$ occurs in an other atom of the body. 

Consider now a set $\Sigma$ of tgds, a tgd $r{:}B {\to} H$ in $\Sigma$, and an atom $A$ which can be simplified in $r$.
Given a tgd $r'{:}B'{\to}H'$ in $\Sigma'$ and an atom $A' \in H'$ we say that $r'$ \emph{unifies} with $(A,r)$ when there 
 exists a renaming $\theta_1$ of $\V A$ and a renaming $\theta_2$ of $\V{B'}\cap \V{H'}$ such that $A[\theta_1] \in H'[\theta_2]$, in which case $(\theta_1,\theta_2)$ is called a \emph{unifier} of 
$r'$ and $(A,r)$. The following algorithm describes the new set of tgds $\Sigma_{A,r}$ that results from the \emph{simplification} of $(A,r)$ in $\Sigma$:      
\begin{center}
\begin{tabular}{l}
 Let $\bar x = (x_1,\ldots,x_n)$ be an ordering of $X_{A,r} = \{\bar x\} $\\
 Let $R_a$ be a fresh predicate of arity $n$\\
 Replace $A$ by $R_a(\bar x)$ in the body of $r$  \\
 For all tgd $r'{:}B'{\to} H'$ in $\Sigma$, including $r'=r$\\
 \qquad   For all unifier $(\theta_1,\theta_2)$ of $r'$ and $(A,r)$ \\
 \qquad\qquad      Add the tgd $B'[\theta_2] \to R_a(\theta_1(\bar x))$.   
 \end{tabular}
 \end{center}
Note that the quantification of the variables may be modified in this process, and new simplifications may therefore be possible in $\Sigma_{A,r}$.  For instance, the set of tgds
$$
\begin{array}{l}
r_1: A(x,x,y,z,t) \to B(x,y)\\
r_2: C(x,y) \to \exists u,v,\ A(x,y,u,v,v)\\
r_3: D(x,y,z,t) \to A(x,x,y,z,t)
\end{array}
$$
will, after a simplification step in $r_1$,  be replaced by 
$$
\begin{array}{l}
r_1': R_a(x,y) \to B(x,y) \\
r_2: C(x,y) \to \exists u,v,\ A(x,y,u,v,v)\\
r_3: D(x,y,z,t) \to A(x,x,y,z,t) \vspace{2mm}\\
r_2': C(x,x) \to \exists u, R_a(x,u)\\
r_3': D(x,y,z,t) \to  R_a(x,y)
\end{array}
$$
and the first atom of $r'_3$ can now be simplified (even though this atom could not be simplified in $r_3$).

Despite the previous observation, we can check that the process of simplification always terminates since each step introduces only a finite number of tgds, and the number of variables in each of these tgds is strictly decreasing. For a finite set $\Sigma$ of tgds, we finally define (non-deterministically) the set $\Sigma{\downarrow}$ as a set of tgds obtained from $\Sigma$ by repeating the operation of simplification until a fixed point is reached.

\begin{theorem} \label{theo:theosimprew}
For all finite sets $\Sigma$ of tgds and all finite sets $\Q$ of instances, there is a finite query rewriting for  $(\Sigma,\Q)$ iff there is a finite query rewriting for $(\Sigma{\downarrow},\Q)$.
\end{theorem}

\begin{proof}[Sketch]
Consider a series of simplification steps $s_1,\ldots,s_n$ where each $s_i$ is characterized by the atom 
$A_i(\bar x_i,\bar y_i)$ that has been simplified at step $s_i$ and the corresponding atom $R_i(\bar x_i)$ that has been introduced. 
For each $i$, consider the tgd $r_{i} : A_i(\bar x_i,\bar y_i) \to R_i(\bar x_i)$. Finally, let $\Gamma = \{r_i\}_{i\in\{1..n\}}$. We can observe that, for all instance $D$ over the original schema, and every data rewriting $D_\Gamma$ for 
 $(\D,\Gamma)$, we have $D \land \Sigma \entails \Q$ iff $D_\Gamma \land \Sigma{\downarrow} \entails Q$.
Since, $\Gamma$ is a set of Datalog rules, there is a data rewriting $D_\Gamma$ which is finite iff  $D$ is finite. 
We can also observe that the set of tgds $\Sigma^{-1} = \{R_i(\bar x_i) \to \exists \bar y_i A_i(\bar x_i,\bar y_i)\}$ is acyclic, and for all instance $D'$ of the extended schema, there is therefore a data rewriting $D'_{\Gamma^{-1}}$ of $(\D',\Sigma^{-1})$ which is finite iff $D'$ is finite. With letting $\Pi$ be the operation that projects an instance of the extended schema on the original schema, we can then observe that, for all instance $D$ of the original schema, we have 
$\Pi((D_{\Gamma})_{\Gamma^{-1}}) \equiv D$. This means intuitively that there exists a one-to-one correspondence (which preserves finiteness) between the instances of the original schema and the instances of the extended schema. This is the key argument behind the proof of Theorem~\ref{theo:theosimprew}.   
\end{proof}  

\noindent {\bf Stickiness (Slightly Revisited)}. Given a set of atoms $B$ and a term $t$, we denote by $\pos(t,B)$ the set of pairs $(R,i)$, called \emph{positions}, such that $B$ contains an atom $R(t_1,\ldots,t_n)$ where $t_i = t$.  Given a set of tgds $\Sigma$, a tgd $B{\to}H$ in $\Sigma$ and an atom $A \in H$, the tgd $B \to A$ is called a \emph{global-as-view projection} of $\Sigma$, denoted $r' \in \GAV(\Sigma)$. Given a set of tgds $\Sigma$, we define the set $\mathcal A_\Sigma$ of \emph{affected} positions as the smallest set of positions such that, for all tgd $r \in \GAV(\Sigma)$ of the form 
$r:B(X,Y) \to H(X,Z)$,
we have: 
\begin{itemize}
\item[(i)] $\forall v \in Y, \pos(v,B) \subseteq \mathcal A_\Sigma$
\item[(ii)] $\forall u \in X, \ (\pos(u,H) \subseteq \mathcal A_\Sigma) \Rightarrow (\pos(u,B) \subseteq \mathcal A_\Sigma)$
\end{itemize} 
We say that $\Sigma$ is \emph{sticky} iff, for all tgd $r \in \GAV(\Sigma)$ of the form above, 
and all $u \in X$ such that $\pos(u,B) \subseteq \mathcal A_\Sigma$, the variable $u$ occurs in only one atom of the body.
This definition  differs from \cite{CGP10} because of this last requirement ``in only one atom". In contrast, the definition from \cite{CGP10} requires that $u$ occurs ``only once" (in only one atom \emph{and} in only one position). 

\begin{theorem}
If $\Sigma$ is a sticky set of tgds, then $\Sigma{\downarrow}$ is a set of lossless tgds, and therefore,  
for all set $\Q$ of instances, there is a finite query rewriting for $(\Sigma,\Q)$.
\end{theorem}

\begin{proof}[Sketch]
It can be checked that, for every sticky set $\Sigma$ of tgds, and every atom $A$ that can be simplified in a tgd  
$r \in \Sigma$, the set $\Sigma_{A,r}$ resulting from the simplification of $(A,r)$ is also a sticky set of tgds. 
Assume now that $\Sigma{\downarrow}$ is sticky and $\Sigma{\downarrow}$ contains a tgd $r:B \to H$ which is not lossless. Since $r$ is not lossless, there is an atom $A_h \in H$ and an atom $A_b \in B$ such that $\V{A_b} \not\subseteq \V{A_h}$. Observe that the tgd $B \to A_h$ belongs to $\GAV(\Sigma{\downarrow})$ and consider a variable $v \in \V{A_b} {\setminus} \V{A_h}$. By definition of $\mathcal A_{\Sigma{\downarrow}}$, this variable $v$ occurs in an affected position, and the stickiness assumption ensures that $v$ occurs only in the atom $A_b$, meaning that $v \not\in \V{B{\setminus} A_b}$. 
If follows that $A_b$ can be simplified in $r$, and this contradicts the definition of $\Sigma{\downarrow}$. Therefore, every tgd in $\Sigma{\downarrow}$ is lossless.  
\end{proof}

We can observe that the (revisited) notion of stickiness is simultaneously a strict generalisation of: (1) the class of lossless tgds; (2) the original notion defined in \cite{CGP10} from which it is inspired; and (3) the class of lav tgds, which was not yet covered by (2). Note also that stickiness could be combined with the (incomparable) notion of acyclicity discussed in Proposition~\ref{prop:termres} and/or the class of \emph{sticky-join} tgds introduced in \cite{CGP10} to design an ever larger class of tractable settings. However, this is left as future work.

\subsection{Integrating the Egds}

This section provides a negative result on egds and query rewriting which will motivate two further contributions: (1) a novel technique, also in this section, that allows the \emph{integrating} of some egds in a set of tgds; and (2) the study, in Section 4, of a richer notion of rewriting based on Datalog. 

While the completeness result from Section 3.1 remains of clear interest with both tgds and egds, it turns out that the notion of query rewriting from Section 3.2 is in fact very limited under egds. Intuitively, this is because we considered a notion of \emph{first-order} rewriting, while dealing with egds often requires the power of second-order (or, as we will see, the use of some integration technique which extends the schema). In fact, as illustrated by the following proposition, $(\Sigma,\Q)$ is rarely rewritable under egds, even in the case where $\Sigma$ consists of a single egd.

\begin{proposition} \label{prop:egdnotrew}
There is no finite query rewriting for
$$
\begin{array}{l}
\Sigma = \{A(x,y),A(x,y') \to y=y'\}, \ \mbox{and}\\
\Q = \{R(z,z)\}.
\end{array}
$$
\end{proposition}

\begin{proof}[Sketch]
An infinite rewriting for $(\Sigma,\Q)$ is the set of instances $\R = \{R_n\}_{n\ge 1}$ where 
each $R_i$ is equal to  $
\{R(x_1,x_n)\} \cup \{ A(x_i,y_{i}),A(x_{i+1},y_{i})\, |\, i \le n-1\}
$. 
We can check that this rewriting $\R$ is not equivalent (up to $\equiv$) to any finite set of instances. Therefore, there is no finite rewriting for $(\Sigma,\Q)$. 
\end{proof}

\newcommand\F{\mathcal F} 

Despite the above result, it has been observed in \cite{myVLDB} that there are practical scenarios where egds can be `handled' with a first order language. It is indeed possible, in some cases, to \emph{integrate} these egds in the given set of tgds, as to compute a new set of tgds which, intuitively, does not interact with these egds. As a special case, this approach based on \emph{integration}, covers the scenarios where the given sets of tgds and egds are already \emph{non-conflicting}, as defined in \cite{JoKl84} or \cite{CaGL09}. However, we will also capture scenarios where the original set of tgds properly interacts with the egds. 

As in \cite{myVLDB} or \cite{CaGL09}, we next focus on \emph{functional dependencies} rather than arbitrary egds. The reason for this is that the egds used in practice often consist of functional dependencies, and the functional dependencies have a more specific syntax which proves more convenient (in the context of integration). 
Recall that a functional dependency is a rule of the form $R_\alpha \lfloor K_\alpha \rfloor \to l_\alpha$ where 
 $R_\alpha$ is a predicate of arity $a_\alpha$, $K_\alpha\subseteq\{1,\ldots,a_\alpha\}$ and $l_\alpha \in \{1,\ldots,a_\alpha\}\setminus K$. Given an instance $M$, we then let $M \entails \alpha$ when, for all atoms of the form 
$R_\alpha(x_1,\ldots,x_{a_\alpha})$ and $R_\alpha(x'_1,\ldots,x'_{a_\alpha})$ in $M$, there either exists some 
 $k \in K_\alpha$ such that $x'_k \not=x_k$, or it holds that $x'_{l_\alpha} = x_{l_\alpha}$. It is clear that a functional dependency can always be expressed by an equivalent egd. For example, for a binary predicate $A$, the dependency $\alpha: A \lfloor 1 \rfloor \to 2$  is equivalent to the egd $A(x,y),A(x,y') \to y=y'$.\\

{\bf Integration Heuristic.} Given a set $\Sigma$ of tgds and a functional dependency $\alpha: R_\alpha \lfloor K_\alpha \rfloor {\to} l_\alpha$, we let $\Sigma_{\alpha}$ be the set of tgds obtained as follows: 
\begin{center}
\begin{tabular}{l}
Start with $\Sigma_{\alpha}:= \Sigma$\\
Let $\mathsf D_\alpha$ and $\mathsf F_\alpha$ be two fresh predicates \\ 
Let $i_1,\ldots,i_n$ be an ordering of $K_\alpha$\\
Add to $\Sigma_{\alpha}$ the two following tgds: \\  
\quad $R_\alpha(x_1,\ldots,x_{a_\alpha}) \to \mathsf F_\alpha(x_{i_1},\ldots,x_{i_n},x_{l_\alpha})$\\ 
\quad $\mathsf D_\alpha(x_1,\ldots,x_n) \to \exists y, \mathsf F_\alpha(x_1,\ldots,x_n,y)$\\
For all tgd $r{:}B \to H$ in $\Sigma$ \\
\quad For all atom $R_\alpha(t_1,\ldots,t_{a_\alpha})$ in $H$\\
\ \ \ {\tiny (*)} If $\{t_{i_1},\ldots,t_{i_n}\} \subseteq \V B$ and $t_{l_\alpha} \not\in \V B$\\
\quad\quad\quad  Add the atom $\mathsf F_{\alpha}(t_{i_1},\ldots,t_{i_n},t_{l_\alpha})$ to the body of $r$\\
\quad\quad\quad  Add the tgd $B \to \mathsf D_\alpha(t_{i_1},\ldots,t_{i_n})$ to $\Sigma_{\alpha}$.\\
\end{tabular}
\end{center}
We say that $\alpha$ \emph{interacts} with $\Sigma$ when the lines below {\tiny (*)} in the above algorithm are actually used. That is, when
 there is a tgd $r{:}B{\to}H$ in $\Sigma$ and an atom $A \in H$ of the form $R_\alpha(t_1,\ldots,t_{a_\alpha})$ 
 such that $\{t_{i} | i \in K_\alpha\} \subseteq \V B$ and $t_{l_\alpha} \not\in \V B$. 
Note in particular that $\alpha$ does not interact with $\Sigma$ when $\alpha$ is \emph{non-conflicting} with $\Sigma$ according to the definition of \cite{CaGL09} (which would requires here that $\{t_{i} | i \in K_\alpha\} \not\subseteq \V B$).

\begin{definition}
We say that the integration of $\alpha$ succeeds in a set of tgds $\Sigma$ iff the set of tgds $\Sigma_\alpha$ is such that: 
\begin{itemize}
\item $\alpha$ does not interact with $\Sigma_\alpha$, and 
\item for all tgd $B'{\to}H'$ in $\Sigma_\alpha$, $B' \entails \mathsf F_\alpha \lfloor 1,\ldots,n \rfloor {\to}(n+1)$.
\end{itemize}
\end{definition}

\begin{lemma} \label{lem:integ}
If the integration of $\alpha$ succeeds in $\Sigma$ then, for all instances $D$ and $Q$ over the original schema such that 
$D \entails \alpha$, the following statements are equivalent:
\begin{itemize}
\item $D \land \Sigma \land \alpha \entails Q$
\item $D \land \Sigma_{\alpha} \entails Q$  
\end{itemize}   
\end{lemma}

\begin{proof}[Sketch]
Let $U \in \Fix(D,\chase_{\Sigma_\alpha})$ and recall from Theorem 2 that $U$ is a data rewriting for $(D,\Sigma_{\alpha})$. Under the assumptions that $D \entails \alpha$ and $\alpha$ does not interact with $\Sigma_{\alpha}$, we can check that $U \entails \alpha$. It follows that $U$ is also a data rewriting for 
$(D,\Sigma_{\alpha} \land \alpha)$. We can finally check that the following statements are all equivalent: $D \land \Sigma_{\alpha} \entails Q$; $U \entails Q$; $D \land \Sigma_{\alpha} \land \alpha \entails Q$; and $D \land \Sigma \land \alpha \entails Q$.
\end{proof}

Consider  now a set of functional dependencies $\mathcal F$ and a set $\Sigma$ of tgds. We say a set of tgd $\Sigma_\F$ of tgds \emph{integrates} $\F$ in $\Sigma$ iff there exists 
a series $\alpha_1,\ldots,\alpha_n \in \F$ and a series 
$\Sigma_{0}\,,\,\Sigma_{1}\,,\,\ldots\,,\,\Sigma_n$ such that: 
\begin{itemize}
\item $\Sigma_0 = \Sigma$, $\forall i\ \Sigma_{i+1} = (\Sigma_{i})_{\alpha_i}$ and $\Sigma_n = \Sigma_\F$; 
\item for all $i$, the integration of $\alpha_i$ succeeds in $\Sigma_{i}$; and 
\item there is no remaining $\alpha \in \F$ that interacts with $\Sigma_\F$.
\end{itemize} 

We are now ready to formalize the property of interest which is ensured by the integration heuristic. 

\begin{theorem} \label{theo:integrate}
Given a set of tgds $\Sigma_\F$ that integrates a set of functional dependencies $\F$ in a set of tgds $\Sigma$, for all instances $D$, all data rewriting $D_{\mathcal \F}$ of $(D,\F)$, and all sets $\Q$ of instances, the following statements are equivalent: 
\begin{itemize}
\item $D \land \Sigma \land \mathcal F \entails \Q$ 
\item $D_{\mathcal F} \land \Sigma_{\mathcal F} \entails \Q$  
\end{itemize}
\end{theorem}

\begin{proof}[Sketch]
The result can be proven by induction on the cardinality of $\F$ using Lemma~\ref{lem:integ} and the result of \emph{separability} which was established in \cite{CaGK08}. 
\end{proof}

Note finally that, since $\F$ is a set of functional dependencies, we can compute a data rewriting $D_{\mathcal F}$ for $(\D,\F)$ in polynomial time (data complexity) using any standard chase procedure. Combining Theorem \ref{theo:integrate} with the results of the previous section, we finally get the following result:  

\begin{corollary}
Given $\Sigma_\F$ that integrates $\F$ in $\Sigma$, if the set $\Sigma_\F$ is sticky, for all set $\Q$ of instances, the following problem is \textsc{Ptime}: given an instance $D$, does $D \land \Sigma \land \F \entails \Q$? 
\end{corollary}

We finally provide  an example of scenario taken from \cite{myVLDB} which is covered by the approach described in this section: 
$$
\begin{array}{l}
\Sigma = \{ A(x,y) \to \exists z, B(x,z) \land C(z,y) \} \vspace{1mm}\\
\mathcal F = \{\alpha: B\lfloor 1 \rfloor \to 2\} \vspace{1mm}\\
\Sigma_{\mathcal F} = \left\{
\begin{array}{@{}l@{}}
B(x,y) \to \mathsf F_\alpha(x,y) \\ 
\mathsf D(x) \to \exists y,\, \mathsf F_\alpha(x,y) \\ 
A(x,y) \to \mathsf D_\alpha(x) \\ 
A(x,y) \land \mathsf F_\alpha(x,z) \to B(x,z) \land C(z,y) \\
\end{array}\right\}
\end{array}
$$

\subsection{Intermezzo: Constants and Free Variables}
\label{sec:refinements}

The goal of this section is to show how the previous results can be applied to more realistic \emph{databases} (with constants and nulls) and \emph{non-boolean queries} (with free variables).

\subsubsection{Hard and Soft Constants}

 A \emph{database} $D$ is a set of atoms $R(t_1,\ldots,t_n)$ where each term $t_i$ is either a \emph{variable} (also known as a \emph{labelled null}) or a constant from a finite set $\Delta = \Delta_{h} \uplus \Delta_{s}$ where: $\Delta_{h}$ is a set of \emph{hard constants} which are subject to the standard \emph{unique name assumption} (UNA); and $\Delta_{s}$ is a set of \emph{soft constant}  which are not subject to the UNA (see e.g. \cite{brunoPODS}). Given  a database $D$ we denote by $D^*$ the instance obtained from $D$ as follows: (1) rename every $c \in \Delta$ into a variable $v_c$; (2) for every $c \in \Delta$ introduce a fresh predicate $R_c$ and add the atom $R_c(v_c)$; (3) introduce a fresh predicate $R_=$ and, for all $c,c' \in \Delta_h$ such that $c\not=c'$, add the atom $R_{\not=}(v_{c},v_{c'})$.  
This definitions correspond to the standard encoding of constants and it can similarly be applied to boolean queries with constants. The following properties are then readily verified: (1) given a database $D$ and a set $\Sigma$ of tgds and egds, $D \land \Sigma$ is satisfiable iff $D^* \land \Sigma \not\entails \{R_{\not=}(x,x)\}$; and (2)  when $D \land \Sigma$ is satisfiable, for all set $\Q$ of instances, we have  
$D \land \Sigma \entails \Q$ iff $D^* \land \Sigma \entails \Q^*$. A less obvious observation, formalized below, is that this technique of simulation can also be used for query rewriting under integrable egds.

\begin{proposition} \ 
\begin{itemize}
\item Given a set of tgds $\Sigma$, an integrable set of functional dependencies $\mathcal F$, and a database $D$, the formula \linebreak  $D \land \Sigma \land \mathcal F$ is satisfiable iff $D^* \land \Sigma_{\mathcal F} \not\entails \{R_{\not=}(x,x)\}$.
\item Under satisfiability, for all databases $D$, all data rewritings  $D^*_{\mathcal F}$ of 
$(D^*,\Sigma_{\mathcal F})$ and all sets $\Q$ of instances, we have $D \land \Sigma \land \mathcal F \entails \Q$ 
iff  $D^*_{\mathcal F} \land \Sigma_{\mathcal F} \entails \Q^*$. 
\end{itemize}    
\end{proposition}

\begin{corollary}
Given a set $\Sigma$ of tgds and an integrable set $\F$ of functional dependencies, if $\Sigma_{\mathcal F}$ is sticky, then, for all set $\Q$ of instances, the following problem is in \textsc{Ptime}: given a database $D$, does $D \land \Sigma \land \mathcal F \entails \Q$?  
\end{corollary}

\subsubsection{Free Variables and Equalities}

Recall that a  query $\Q$ is called a \emph{union of conjunctive queries with equalities}, denoted $Q \in \UCQ^=$, when $Q$ is a first-order query of the form 
$$\Q = \{(x_1,\ldots,x_a)| \mbox{${\bigvee}_j$} Q_j\}$$
 where each $x_i$ is a called a \emph{free variable} and each clause $Q_j$ is a finite conjunction of relational atoms and equality atoms (with constants, free variables, and existential variables). Given such a query $\Q$, we may consider the set $\Q^*$ of instances obtained as follows: (1) rename every constant $c$ into a variable $v_c$ and add the atom $R_c(v_c)$ to each clause; (2) for every free variable $x_i$, introduce a predicate $V_i$ and add the atom $V_i(x_i)$ to each clause; (3) introduce a predicate $R_=$ and replace every equality atom $(t = t')$ by the atom $R_=(t,t')$; and (4) let $Q^*$ be the resulting set of clauses (which now consist of instances with variables only). Conversely, given an instances $R^*$ of the extended schema (and for a fixed tuple $(x_1,\ldots,x_n)$ of free variable) we let $R$ the set of relational and equality atoms over the original schema which is obtained by replacing every atom $R_c(u)$ by $(u=c)$ and every atom $V_i(u)$ by $u = x_i$. Given a set $\R^*$ of instances over the extended schema, we finally denote by $\R$ the $\UCQ^=$ query of the form 
$
\R = \{ (x_1,\ldots,x_a) | \bigvee\{ R | R^* \in \R^* \}\}
$.
These definitions are exemplified below:
$$
\begin{array}{l@{\,}l}
\Sigma = & 
\{A(u,v) \to B(u,u)\}\\
\Q = & \{(x_1,x_2)\, |\, B(x_1,x_2) \} \vspace{2mm}\\
\Q^* = & \{\{B(x_1,x_2), V_1(x_1),V_2(x_2)\}\}\\
\R^* = & \Q^* \cup \{\{A(u,v), V_1(u),V_2(u)\}\} \vspace{2mm}\\
\R = & \{(x_1,x_2)\,|\, B(x_1,x_2) \lor (\exists u,v,\, A(u,v) {\land} x_1{=}u {\land} x_2{=}u)\}.  \\
\end{array}
$$

As next formalized, the above technique can be used to generalize the results of the previous sections to the case of non-boolean queries with constants and equalities: 

\begin{proposition} \label{prop:4}
Given $\Sigma$, $\Q$, $\Q^*$,$\R^*$,$\R$ as above where  
$$\R^* \in \Fix(\Q^*,\res_\Sigma)$$
the $\UCQ^=$ query $\R$ is a rewriting of the $\UCQ^=$ query $\Q$. More formally, for all database 
$D$ and all tuples $\bar c$ of constants, the following statements are equivalent:
\begin{itemize}
\item $\bar c$ is an answer of $\R$ in $D$, denoted $\bar c \in \R(D)$
\item $\bar c$ is a certain answer of $\Q$ in $D$ under $\Sigma$, meaning that $\bar c \in Q(D')$ for all instance $D'$ such that $D \land \Sigma \entails D'$.
\end{itemize} 
\end{proposition}

\begin{corollary}
We can use the technique of saturated resolution to compute a finite $\UCQ^=$ rewriting for a given $\UCQ^=$ query whenever there exists one (e.g. under stickiness).  
\end{corollary}

\section{Datalog Rewriting}

As announced in the introduction, we now consider a more general notion of rewriting, called \emph{Datalog rewriting}, which will prove to be a unifying paradigm of tractability for Datalog$^\pm$. More precisely, we will show that it captures the class of terminating dependencies (Section 4.1) and the class of guarded tgds (Section 4.2).

\begin{definition} 
Given a set $\Sigma$ of tgds and egds, and a set $\Q$ of instances, a Datalog rewriting for $(\Sigma,\Q)$ is a triple $(\A,\Gamma,G)$ where
\begin{itemize}
\item $\A$ is a set of predicates which do not occur in $\Sigma$ or $\Q$;
\item $\Gamma$ is a finite set of tgds $B{\to}H$ where $\V H \subseteq \V B$;
\item $G$ is an instance of the form $G=\{\mathsf{G}()\}$ where $\mathsf{G} \in \A$; 
\item for all $\A$-free instances $D$, it holds that: 
$$
D \land \Sigma \entails Q \quad \mbox{iff} 
\quad 
D \land \Gamma \entails G.
$$
\end{itemize}
\end{definition}

Note that, in the above definition, each tgd in $\Gamma$ correspond to a standard Datalog rule (also known as a \emph{full tgd}). Since $\Gamma$ is required to be finite, $\Gamma$ corresponds to a standard Datalog program. A predicate of $\A$ will be called an \emph{auxiliary} predicate and $\A$ corresponds intuitively to an \emph{intentional} schema. 
An $\A$-\emph{free instance} is defined as an instance in which no predicate of $\A$ occurs. That is, an $\A$-free instance corresponds intuitively to an \emph{extensional} database. The instance $G$ is finally known as the \emph{goal} of the Datalog program $(\A,\Gamma,G)$ and the predicate $\mathsf G$, of arity 0, is known as the \emph{goal predicate}. The following proposition finally summarizes the basic  properties of Datalog rewritings: 

\begin{proposition}   \ 
\begin{itemize}
\item If a Datalog rewriting exists for $(\Sigma,\Q)$, the following problem is in \textsc{Ptime}: given $D$, does $\	D \land \Sigma \entails \Q$? 
\item If there is a finite query rewriting for $(\Sigma,\Q)$ then there is also a Datalog rewriting for $(\Sigma,\Q)$.
\item There are some pairs $(\Sigma,\Q)$ for which a Datalog rewriting exists while no finite query rewriting exists. 
\end{itemize}
\end{proposition}

\begin{proof}[Sketch]
The first point follows from the following observation: when $\Gamma$ is a set of full tgds, we can compute a data rewriting $\U$ for $(D,\Gamma)$ in polynomial time (for a fixed $\Gamma$) using the chase, and we can then test in polynomial time (for a fixed instance $G$) whether $\U \entails G$. For the second point, 
given a finite query rewriting $\R$ for $(\Sigma,\Q)$ and with letting  
$G = \{\mathsf{G}()\}$ for some fresh predicate $\mathsf{G}$, we can observe that $(\{\mathsf G\},\{R \to G\}_{R \in \R},G)$ is a Datalog rewriting for $(\Sigma,\Q)$. Finally, for $\Sigma{=}\{R(x,y),R(y,z){\to}R(x,z)\}$ and $\Q=\{R(x,x)\}$, the pair 
$(\{\mathsf G\}, \Sigma \cup \{R(x,x) \to G\},G)$ is  a Datalog rewriting of $(\Sigma,\Q)$ while there is no finite (first-order) query rewriting for $(\Sigma,\Q)$. 
\end{proof}

\subsection{From Termination To Datalog}

This section revisits the criterion of \emph{oblivious termination} which was introduced in \cite{brunoPODS} and 
 presented in \cite{CGLMP10} as a language of the Datalog$^\pm$ family. As discussed in \cite{CGLMP10}, there are alternative criteria of termination that can be considered (see \cite{Termin} for the current the state of the art). Note however that oblivious termination captures the case of \emph{weakly acyclic} \cite{FKMP05} sets of tgds and arbitrary sets  of egds, and the results presented in this section can be extended to the classes discussed in \cite{Termin}.  In a nutshell, oblivious termination is based on (1) a technique of simulation that encodes the egds by means of tgds and (2) a standard notion of skolemization which generates a set of rules with function symbols (that is, a \emph{logic program}). As observed in \cite{brunoPODS}, this logic program enjoys a technical \emph{bounded depth} property. In turn, we will show in this section that this bounded depth property ensures the existence of a Datalog rewriting. \\

\newcommand\E{{\mathsf E}}

{\bf Simulation.}  Given a set $\Sigma$ of tgds and egds, we say that $\Sigma'$ is a \emph{substitution-free} simulation of $\Sigma$, denoted $\Sigma' \in \mathsf{Sim}_\E(\Sigma)$, when $\Sigma'$ is a set of tgds obtained from $\Sigma$ using the simulation technique from \cite{brunoPODS} (which, unlike alternative techniques, avoid the use of substitution axioms). More precisely, we let $\Sigma' \in \mathsf{Sim}_\E(\Sigma)$ when $\E$ is a binary predicate which does not occur in $\Sigma$ and $\Sigma'$ can be computed with the following non-deterministic algorithm:  
\begin{center}
\begin{tabular}{l}
Start from $\Sigma'=\Sigma$.\\
Add the following tgds to $\Sigma'$: \\
\qquad  $\E(x,y) \to \E(y,x)$ \\
\qquad  $\E(x,y),\E(y,z) \to \E(x,z)$ \\
For all predicates $R$ occurring in $\Sigma$ (of arity $n$) \\
\quad Add the following tgd to $\Sigma'$:\\ 
\qquad  $R(x_1,\ldots,x_n) \to \E(x_1,x_1),\ldots,\E(x_n,x_n)$ \\
Repeat (until fixed point) \\
\quad If there is a tgd $B{\to}H$ or an egd $B{\to}x{=}y$ in $\Sigma'$ and\\
\quad  a variable $x \in \V B$ that occurs more than once in $B$\\
\quad\quad Let $x'$ be a fresh variable \\
\quad\quad Replace one occurrence of $x$ by $x'$ in $B$\\
\quad\quad Add the atom $\E(x,x')$ to $B$\\   
For all egds $r:B \to x=y$ in $\Sigma'$\\
\quad Replace $r$ by the tgd $B \to \E(x,y)$\\ 
Return $\Sigma'$.
\end{tabular}
\end{center}

{\bf Skolemization.} Given a set of tgds $\Sigma$, we denote by $P_\Sigma$ the logic program which is obtained by skolemizing $\Sigma$ is a standard way. That is,  for all $B(X,Y)\to H(X,Z)$ in $\Sigma$, the program $P_\Sigma$ contains a rule $B(X,Y){\to}H'(X)$ where $H'$ is obtained from $H$ by replacing every variable $z \in Z$ by a term $f(\bar x)$ where $f$ is fresh function symbol and the tuple $\bar x$ is a fixed ordering of $X = \{\bar x\}$. Given an instance $I$ and such a logic program $P_\Sigma$ we then denote by $P_\Sigma(I)$ the fixed point of $I$ and $P_\Sigma$ (also known as the minimal Herbrand model).

\begin{definition}  \ 
\begin{itemize}
\item We say that a set $\Sigma$ of tgds \emph{ensures oblivious termination} iff, for all finite instance $I$, $P_\Sigma(I)$ is finite.
\item Given a set $\Sigma$ of tgds and egds, we say that $\Sigma$ ensures oblivious termination iff there exists $\Sigma' \in \mathsf{Sim}_\E(\Sigma)$ such that $\Sigma'$ ensures oblivious termination.   
\end{itemize}
\end{definition}

{\bf Bounded Depth.} Given a skolem term $t$ with variables and function symbols, we define the \emph{depth} $d(t)$ of $t$ in a standard way. More precisely, given a variable $x$ we let $d(x) = 1$, and given a term $t=f\ang{t_1,\ldots,t_n}$, we let $d(t) = 1 + \mathsf{max}_{i \le n}\big(d(t_i)\big)$. Given a set of tgds $\Sigma$, an instance $D$ and an integer $k$ we denote by $P_\Sigma^k(D)$ the set of atoms with skolem terms which is obtained by applying (inductively) all the rules $B(X,Y) \to H'(X)$
 of $P_\Sigma$, but only for the valuations $\theta$ of $X \cup Y$ such that, for all $u \in X\cup Y$, the depth of $\theta(u)$ is at most $k$. It can be checked that $P_\Sigma^k(I)$ is finite and well-defined whenever $\Sigma$, $I$ and $k$ are finite. In particular, the order of application of the rules does not matter. Note here that $P^k_\Sigma(I)$ is a skolem instance (with skolem terms) rather than a standard instance (with variables only) but he definitions from Section 3.1 can nonetheless be extended in a natural way. In particular, given a set $\Q$ of instances , we let $P^k_\Sigma(I) \entails Q$ iff, for all $Q \in \Q$, there is a mapping $\theta$ from $\V Q$ to the terms of $P^k_\Sigma(I)$ such that $Q[\theta] \subseteq P^k_\Sigma(I)$. 

\begin{definition}
Given a set $\Sigma$ of tgds and a set $\Q$ of instances, we say that $(\Sigma,\Q)$ has \emph{bounded depth} iff there exists a finite integer $k$ (depending only on $(\Sigma,Q)$) such that, for all instances $D$, the following statements are equivalent: 
\begin{itemize}
\item $D \land \Sigma \entails \Q$ 
\item $P^k_\Sigma(D) \entails \Q$
\end{itemize}   
\end{definition}

The following result was established in \cite{brunoPODS}:

\begin{lemma}
If $\Sigma$ is a finite set of tgds ensuring oblivious termination, there exists $k$ (depending only on $\Sigma$) such that, for all instances $D$, $P_\Sigma(D)= P^k_\Sigma(D)$. 
Therefore, for all sets $\Q$ of instances, $(\Sigma,\Q)$ has bounded depth.  
\end{lemma}

We next present the main result of this section.

\begin{theorem} \label{theo:bdepth}
Given a set $\Sigma$ of tgds and a set $\Q$ of instances, if $(\Sigma,\Q)$ has bounded depth (and if we know the bound $k$), we can compute a Datalog rewriting for $(\Sigma,\Q)$. 
\end{theorem}

\begin{proof}[Sketch]
The key idea of the proof is to use fresh predicate symbols to simulate the effect of the function symbols of $P^k_\Sigma$. Intuitively, every atom $R(\bar t)$ with skolem terms can indeed be simulated with a standard atom $R_s(\bar x)$, with variables only, where $s$ encodes the ``shape" of each term, while $\bar x$ corresponds to the variables that occur in $\bar t$. For instance:
$$
R(x,f\ang{x,y},g\ang{y,f\ang{x,z}})  \approx  R_{1,f\ang{1,2},g\ang{2,f\ang{1,3}}}(x,y,z) \\
$$
For a fixed bound $k$, for every rule $B(X,Y) \to H'(X)$ in $P_\Sigma$ and for every valuation $\theta$ of $X \cup Y$ such that $\max \{d(\theta(v)) | v \in X \cup Y\} \le k$ we can then translate every skolem atom in the rule $B[\theta] \to H[\theta]$ into a standard atom. For instance, if $k = 2$ and $P_\Sigma$ contains only one binary skolem function $f$, we can replace the rule 
$$
A(x,y) \to B(x,y,f(x,y))
$$   
by a set of Datalog rules containing, among others, the rules:
$$
\begin{array}{l} 
A(x,y) \to B_{1,2,f\ang{1,2}}(x,y)\\
A_{1,f\ang{1,1}}(x) \to B_{1,f\ang{1,1},f\ang{1,f\ang{1,1}}}(x)\\
A_{1,f\ang{1,2}}(x,y) \to B_{1,f\ang{1,2},f\ang{1,f\ang{1,2}}}(x,y)\\
A_{f\ang{1,2},f\ang{1,3}}(x,y,z) \to B_{f_{1,2},f\ang{1,3},f\ang{f_{1,2},f\ang{1,3}}}(x,y)\\
\ldots
\end{array}
$$
Consider now $\Q$  of the form $\Q=\{\mathsf G()\}$ such that $(\Sigma,\Q)$ has bounded depth $k$. Let $\A^k$ be the set predicates $R_s$ where $s$ encodes a shape of depth $\le k$. Let $\Gamma^k$ be set of Datalog rules resulting from the above construction. We can check in this case that $(\A,\Sigma^k,\{\mathsf G()\})$ is a Datalog rewriting for $(\Sigma,\Q)$. In the general case, (when $\Q$ is not already of the form $\{\mathsf G()\}$), 
we may introduce a fresh 0-ary predicate $G$, consider all the tgds $r_Q : Q \to G()$ where $Q \in \Q$, and consider all the valuations $\theta$ of $X_Q$ such that 
$\mathsf{max}\{d(\theta(x))| x \in X_Q \} \le k$. We can then encode each of these rules with a Datalog rule over $\A^k \cup \{\mathsf G\}$, and we can conclude as in the previous case. 
\end{proof}

\begin{corollary}
There is an algorithm that, given a set $\Sigma$ of tgds and egds ensuring oblivious termination, and given a set $\Q$ of instances, computes a Datalog rewriting for $(\Sigma,\Q)$.  
\end{corollary}

\subsection{From Guardedness To Datalog} 

In this section, we consider the class of \emph{guarded} tgds. Intuitively, we will say that a tgd is guarded when there is an atom in the body (called a \emph{guard}) that contains all the \emph{universal} variables. In turn, a variable is called universal iff it occurs both in the body and the head. Note that the variables that occur only in the body are not taken into account in this definition of guardedness. Therefore, the class of guarded tgds contains, as a special case, the tgds $B \to \mathsf{G}()$ where $B$ is an arbitrary instance and $\mathsf{G}$ is a 0-ary predicate (for example, a goal predicate) because such tgds have no universal variable.  
Another example of guarded tgd is
$$
A(x,y,z),B(i,x,y),B(j,y,z) \to \exists k, B(k,x,z)
$$ 
where the set of universal variables is $\{x,z\}$ and the atom $A(x,y,z)$ is a guard. In contrast, the tgd  
$$
B(i,x,y),B(j,y,z) \to \exists k, B(k,x,z)
$$ 
in \emph{not} guarded since there is no atom in the body that contains both $x$ and $z$. 

\begin{definition}
A tgd of $B \to H$ is \emph{guarded} iff there is a atom $G \in B$ such that $(\V B \cap \V H) \subseteq \V G$.
\end{definition}

\subsubsection{The Case of $\beta$-Guardedness}

A special class of guarded tgds was considered in \cite{CaGK08,CGLMP10}, called $\beta$-guarded tgds in this section, that complies with the following definition:

\begin{definition}
A tgd $B \to H$ is called $\beta$-guarded iff there exists an atom $G\in B$ such $\V B \subseteq \V G$. 
\end{definition}

Note that a $\beta$-guarded tgd is (only) a special case of guarded tgds since the requirement $\V B \subseteq \V G$ is stronger than $(\V B \cap \V H) \subseteq \V G$. For example, the tgd 
$$
A(x,y,z),B(i,x,y),B(j,y,z) \to \exists k, B(k,x,z)
$$
is  guarded but not $\beta$-guarded. While $\beta$-guardedness is slightly less general, it was shown in \cite{CaGK08} that the class of $\beta$-guarded tgds remains a very natural class to consider. In particular, it was shown in \cite{CaGK08} that we only need $\beta$-guardedness (as opposed to general guardedness) to cover interesting classes of ontologies, including languages from the DL-lite family \cite{ACDL*05}. Another important advantage of $\beta$-guardedness is that it ensures an useful property, established in \cite{CaGK08}, called the \emph{bounded guard-depth property}. This property (defined for $\beta$-guarded tgds only) proves indeed very relevant here as it coincides in fact with the general property of \emph{bounded depth} which was discussed in the previous section. Combining the results in \cite{CaGK08} with this observation, we obtain the following results:

\begin{lemma} \label{lem:gdep} 
If $\Sigma$ is a finite set of $\beta$-guarded tgds and $\Q$ is a finite set of instance, then $(\Sigma,\Q)$ has bounded depth $k$ for some computable $k$ that depends both on $\Sigma$ and $\Q$.
\end{lemma}

\begin{corollary}[of Theorem~\ref{theo:bdepth}]
There is an algorithm that, given a set $\Sigma$ of $\beta$-guarded tgds and a set $\Q$ of instances, computes a Datalog rewriting for $(\Sigma,\Q)$.  
\end{corollary}
 
As already discussed, a $\beta$-guarded tgd is only a special case of guarded tgd, and the work of \cite{BaGO10} suggests that there is no trivial reduction from the class of guarded tgds to the class of $\beta$-guarded tgds. This is why we consider an alternative approach, in the following section, for the more general case. 

\subsubsection{From Guardedness To Flatness}

In this section, we provide a proof (sketch), based on a technical notion of \emph{flatness}, for the following result:

\begin{theorem} \label{theo:gflat}
For every finite set $\Sigma$ of guarded tgds and set $\Q$ of instances, there is a Datalog rewriting for $(\Sigma,\Q)$.  
\end{theorem}

The key idea of the proof can be summarized as follows: when $\Sigma$ is a set of guarded tgds, there exists an equivalent set of tgds $\Sigma'$ which enjoys the \emph{flat chase property}. This property meaning intuitively that it is sufficient  to chase the tgds $B(X,Y){\to}H(X,Z)$ of $\Sigma'$ for the renamings of $X$ that map $X$ to the variables of the original instance $D$ (which intuitively correspond to constant values). In turn, the flat chase property can be linked with the property of bounded depth (for the depth $k=1$) and Theorem~\ref{theo:bdepth} can be used again to prove the existence of a Datalog rewriting. \\

{\bf Flat Chase.} Consider a triple $(D,I,J)$ of instances and a set $\Sigma$ of tgds. 
When $I \chase_\Sigma\ J$ we say that the chase step is \emph{flat} with respect to $D$, denoted $I \flat{D}_{\Sigma}\ J$, when there is a tgd $B(X,Y){\to}H(X,Z)$ in $\Sigma$ and two substitutions $\theta_X$ and $\theta_Y$ such that:
\begin{itemize}
\item $B[\theta_X,\theta_Y] \subseteq I$ and $J = I \cup H[\theta_X,\theta_Z]$ for some $I$-fresh renaming $\theta_Z$; and
\item in addition, for all $x\in  X$, it holds that $\theta_X(x) \in \V{D}$.
\end{itemize} 
Given an instance $U$, we finally let $U \in \Flat(D,\Sigma)$ when there exists a finite derivation of the form
$$
D = I_1 \flat{D}_{\Sigma}\ I_2 \flat{D}_{\Sigma}  \ldots \flat{D}_{\Sigma}\  I_n = U.  
$$

\begin{definition}
Given a set $\Sigma$ of tgds and a set $\Q$ of instances, we say that $(\Sigma,\Q)$ has the 
\emph{flat chase property} iff, for all instance $D$ the following statements are equivalent: 
\begin{itemize}
\item $D \land \Sigma \entails Q$
\item $\exists U \in \Flat(D,\Sigma),\ U \entails Q$
\end{itemize}  
\end{definition}

\begin{lemma}
If $(\Sigma,\Q)$ has the flat chase property, then $(\Sigma,\Q)$ has bounded depth (for the depth $k=1$) and therefore, there exists a Datalog rewriting for $(\Sigma,\Q)$. 
\end{lemma}

The proof of Theorem~\ref{theo:gflat} finally relies on Lemma~\ref{lem:5} below.

\begin{lemma} \label{lem:5}
For all finite set $\Sigma$ of guarded tgds and all finite set $\Q$ of instances, there exists a finite set $\Sigma'$ of tgds such that $\Sigma \equiv \Sigma'$ and $(\Sigma',\Q)$ has the flat chase property. 
\end{lemma}

\begin{proof}[Sketch]
 Given a tgd $r: B(X,Y){\to}H(X,Z)$, a refinement of $r$ is a tgd $r'=B'{\to}H'$ where $B'=B[\theta_X,\theta_Y]$ and $H[\theta_X,\theta_Z]$ for some renamings $\theta_X$, $\theta_Y$ ,$\theta_Z$ where $\theta_Z$ is a $B'$-fresh renaming of $Z$. We say that $r'$ is a \emph{careful} refinement of $r$ when, in addition, $\theta_Y$ is a $H'$-fresh renaming. Given a set of tgds $\Sigma$, we let $\Ref(\Sigma)$ (resp. $\CRef(\Sigma)$) the sets of tgds corresponding to a refinement (resp. careful refinement) of some tgd of $\Sigma$. Given a tgd $r:B \to H$ we say that $r$ is of the \emph{split form}   
$$
r:G(X,U),B'(X',U',V) \to H(X,Z) 
$$
when $X$ and $Z$ are defined as usual, 
$G$ is a subset of $B$ satisfying $X \subseteq \V{G}$, 
 $U$ is the set of remaining variables in $G$, $B'$ is the set of atoms in $B{\setminus}G$, $V$ is the set of variables 
 that occur only in $B'$, and finally $(X',Y')= (X \cap \V{B'},Y \cap \V{B'})$.

Given two sets $\Sigma_1$ and $\Sigma_2$ of tgds, we say that a tgd $r_{3}$ is derived from $(\Sigma_1,\Sigma_2)$ when there exists a tgd $r_1 \in \CRef(\Sigma_1)$ and a tgd $r_2 \in \Ref(\Sigma_2)$ such that:
\begin{itemize}
\item $r_1: B_1(X_1,Y_1) \to H_1(X_1,Z_1)$; 
\item $r_2: G_2(X_2,U_2), B'_2(X'_2,U'_2,V_2) \to H_2(X_2,Z_2)$\vspace{2mm}\\ 
\begin{tabular}{ll}
where : &1.\  $G_2$ is non-empty and contained in $H_1$  \\
&2.  $(X'_2 \cup U'_2 \cup V_2)$ is disjoint from $(Z_1 \cup Y_1)$ \\
&3.  $Z_2$ is disjoint from $(X_1 {\cup} Y_1 {\cup} Z_1)$; and 
\end{tabular}
\item $r_{3} = (B_1 \cup B'_2) \to (H_1 \cup H_2)$. 
\end{itemize}
It follows here from 1. and 2. that $X'_2 \cup U'_2 \subseteq X_1$ while 
 $V_2$ is disjoint $Y_1$. Therefore, the tgd $r_{3}$ is of the slit form   
$$
r_3 : G_3(X_3,U_3), B'_3(X'_3,\emptyset,V_3) \to H_3(X_3,Z_3)
$$
where $G_3 = B_1$, $X_3 = X_1$, $U_3 = Y_1$, $B'_3 = B'_2$, $X'_3 = X'_2 \cup U'_2$, $V_3 = V_2$, $H_3 = H_1 \cup H_2$, and $Z_3 = Z_1 \cup Z_2$. 
\setlength{\leftmargini}{9pt}
\begin{itemize}
\item[] {\bf Fact 1.} When a tgd $r_{3}$ is derived from $(\Sigma_1,\Sigma_2)$, it holds that $\Sigma_1 \cup \Sigma_2 \entails r_{3}$
\item[] {\emph{Proof Sketch.}}
Using the previous notations, we have
$$
(B_1 \cup B'_2) \chase_{\{r_1\}} (B_1 \cup B'_2 \cup H_1) \chase_{\{r_2\}} (B_1 \cup B'_2 \cup H_1 \cup H_2) 
$$
and it follows that $\{r_1,r_2\} \entails r_{3}$. 
\end{itemize}
Given an integer $k$ we say that a tgd $r$ is $k$-guarded if $r$ is of the form 
$$
r:G(X,U),B'(X',U',V) \to H(X,Z) 
$$
for some instance $G$ with only one atom (that is, a guard atom), and there exists a set of instances $(B^i)_{i \le n}$  called a $(G,k)$-\emph{decomposition} of $B$  such that:(1) $B' = \cup_{i\le n} (B^i)$, (2) for each $i \le n$, $B^i$ has at most $(k-1)$ atoms; and (3) for all $i \not=j$, it holds that $\V{B^i} \cap \V{B^j} \subseteq \V{G}$.  We define the guard width of a guarded tgd $r$, denoted $\gw(r)$ as the smallest $k$ such that $r$ is $k$-guarded. In contrast, we define the left width of a tgd $r$, denoted $\lw(r)$, as the number of atom in the body of $r$. Note that, every guarded tgd is such that $\gw(r) \le \lw(r)$. Given a set of tgds $\Sigma$, we finally let $\gw(\Sigma) = \mathsf{max}\{\gw(r), r \in \Sigma\}$ and $\lw(\Sigma) = \mathsf{max}\{\gw(r), r \in \Sigma\}$.
\begin{itemize}
\item[] {\bf Fact 2.} 
When a tgd $r_{3}$ is derived from $(\Sigma_1,\Sigma_2)$, it holds that $\gw(r_3) \le \mathsf{max}(\gw(\Sigma_1),\lw(\Sigma_2))$. 
\item[] {\emph{Proof Sketch.}} It can be checked that $\lw(r') \le \lw(r)$ whenever $r'$ is a refinement of $r$ and $\gw(r') \le \gw(r)$ whenever $r'$ is a careful refinement of $r$. Let $r_1 \in \CRef(\Sigma_1)$ be a $k$-guarded tgd of form 
$$
r_1: G_1(X_1,U_1), B'_1(X'_1,U'_1,V_1) \to H_1(X_1,Z_1)
$$
and a $(G_1,k)$-decomposition $B'_1 = \cup_{i\le n} B^i_1$.  
Let $r_2 \in \Ref(\Sigma_2)$ such that $\lw(r_2) \le k$ and $r_2$ is of the form 
$$
r_2: G_2(X_2,U_2), B'_2(X'_2,U'_2,V_2) \to H_2(X_2,Z_2).
$$
Finally, let $r_3$ be the tgd derived from $r_1$ and $r_2$, and recall that $r_3$ is of the form 
$$
r_3 : G_3(X_3,U_3), B'_3(X'_3,\emptyset,V_3) \to H_3(X_3,Z_3)
$$
where $G_3 = G_1 \cup B'$, $X_3 = X_1$, $U_3 = U_1 \cup V_1$, and $B'_3 = B'_2$. 
Since $\lw(r_2) \le k$ and $G_2$ is non-empty, we have $|B'_2|\le (k-1)$. For all $i \le n$, we have
 $\V{B^i_1} \cap \V{B'_2} \subseteq X_1 \subseteq \V{G_1}$. Therefore, the set of instances $\{B^i_1\}_{i\le n} \cup \{B'_2\}$ is a $(G_1,k)$ decomposition of the body of $r_3$ and $r_3$ is $k$-guarded.
\end{itemize}  
Given a set of tgds $\Sigma$, we define the flatening of $\Sigma$, denoted $\Sigma^{\infty}$, as the minimal set of tgds $\Sigma'$ such that  $\Sigma \subseteq \Sigma'$ and $\Sigma'$ contains all the tgds that can be derived from $(\Sigma',\Sigma)$.  
\begin{itemize}
\item[] {\bf Fact 3.} For all set $\Sigma$ of tgds and all set $\Q$ of instances, $(\Sigma^{\infty},\Q)$ has the flat chase property. 
\item[] {\emph{Proof Sketch.}} 
Consider an instance $D$ and suppose that 
$D \land \Sigma \entails \Q$. By completeness of the chase, we know that there exists a derivation
$$
D = I_0 \chase_\Sigma\ I_1 \chase_\Sigma\ \ldots \chase_\Sigma\ I_n 
$$ 
such that $I_n \entails \Q$.  If this derivation is flat, and since $\Sigma \subseteq \Sigma^{\infty}$, 
we have $I_n \in \Flat(D,\Sigma^{\infty})$ and $I_n \entails \Q$. 
Suppose now that the derivation is not flat and consider the first integer $j \le n$ such that 
 $I_{j}$ is obtained from $I_{j-1}$ with a chase step which is not flat. We can the check that there exists a tgd $r_2: 
B_2(X_2,Y_2) \to H_2(X_2,Z_2)$  in $\Ref(\Sigma_2)$ such that 
$B_2 \subseteq I_{j-1}$, $I_{j}= I_{j-1} \cup H_2$ and $X_2 \not\subseteq \V D$. Consider a guard atom $\mathsf{G}(\bar v)$ for $r_2$. Since $X \subseteq \{\bar v\}$ and $X \not \subseteq \Var (D)$, there exits some $i \le j$ such that 
$\mathsf{G}(\bar v) \in (I_{i}\setminus I_{i-1})$. Let $G_2$ be the set of all the atoms of $B_2$ that belong to $(I_{j}\setminus I_{j-1}$. Write $r_2$ under the form  
$$
r_2: G_2(X_2,U_2), B'_2(X'_2,U'_2,V_2) \to H_2(X_2,Z_2).
$$
and observe that $G_2$ is non-empty (since it contains $\mathsf{G}(\bar v)$). Considering $N_j = \V{I_{j}\setminus I_{j-1}}$, since the derivation is flat from the step $1$ to the step $j$, we can observe that the only atom of $I_{i-1}$ containing a variable of $N_j$ are the atoms of $I_{j}{\setminus} I_{j-1}$. Therefore, $(U'_2 \cup V_2)$ is disjoint from $N_j$. 
By definition of careful refinements, we can now consider a tgd $r_1 \in \CRef(\Sigma)$ of the form  
$$
r_1: B_1(X_1,Y_1) \to H_1(X_1,Z_1)
$$
where $Y_1$ is set of variables disjoint from $(X'_2,U'_2,V_2)$, $B_1[\theta_{Y_1}] \subseteq I_{i-1}$ for some renaming $\theta_{Y_1}$ of $Y_1$ and $I_{i} = I_{i-1} \cup H_1$. Let $r_3 = B_3 \to H_3$ be the tgd such that 
$B_3 = B_1 \cup B'_2$ and $H_3 = H_1 \cup H_2$, and observe that $r_3 \in \Sigma^{\infty}$. 
Let $\theta'$ be an injective renaming of $Z_1 \cup Z_2$ into a set of fresh variables (disjoint from $\V{I_n}$) and for all $k \le j$, let $I'_k = I_k[\theta']$. We can observe that we have
$$
D = I_0  \chase_\Sigma\ \ldots \chase_\Sigma\ I_{j-1} \chase_{r_3} I'_j \chase_\Sigma \ldots \chase_\Sigma\ I'_n
$$ 
where the step $I_{j-1} \chase_{r_3} I'_j$ is now a flat step. 

We can finally generalize the above construction to show (by induction on $n$) that every chase derivation of $D$ by $\Sigma$ (leading to $I_n \entails \Q$) can be transformed into a flat derivation of $D$ by $\Sigma^{\infty}$ (leading to $I'_n \entails \Q$)
\end{itemize}
Given a tgd $r:B(X,Y){\to}H(X,Z)$, we define a left-core projection of $r$ as a minimal (for $\subseteq$) set of atoms $B' \subseteq B$ such that $B' = B[\theta_Y]$ for some renaming $\theta_Y$ of $Y$. Given a set of tgds $\Sigma$, we let 
$LC(\Sigma)$ be the set of instances $B'$ such that $B'$ is a left-core projection of some tgd $r \in \Sigma$. 
\begin{itemize}
\item[] {\bf Fact 4.} Given a finite schema $\sigma_0$, an integer $k_0$, and infinite set $\Sigma$ of $k_0$-guarded tgds over $\sigma_0$, the set 
$LC(\Sigma)$ is finite up to isomorphism. 
\item[] {\emph{Proof Sketch.}}  On a fixed schema $\sigma_0$, there are (up to isomorphism) only a finite number of atoms that can be used as a guard $G$ for a decomposition $(B^i)_{i \le n}$. For a fixed integer $k_0$ and for every instance $B = G \cup \bigcup_i B^i$ corresponding to the left-core projection of some tgd, there 
 is only a finite number of possible blocks $B^i$ satisfying $|B^i|\le {k-1}$ that can be used in the composition, up to bijective renaming of $\V{B_i}\setminus G$. It follows that each $B \in LC(\Sigma)$ is finite and that $LC(\Sigma)$ is finite up to isomorphism. 
\end{itemize}
Given a set of tgds $\Sigma$ and an integer $k$, we that that a tgd $r:B \to H$ is a $k$-tgd of $\Sigma$, denoted $\Sigma^{(k)}$ when there exists a tgd $B'\to H'$ in the flatening $\Sigma^\infty$ of $\Sigma$ such that: (1) $H$ is a subset of $H'$ containing at most $k$ atoms;  and (2) $B$ is a left-core projection of $B' \to H$. As a corollary of Fact 4, it can be observed that $\Sigma^{(k)}$ is finite (up to isomorphism) whenever $\Sigma$ is a finite set of guarded tgds. 
\begin{itemize}
\item[] {\bf Fact 5.}
For all sets $\Sigma$ of guarded tgds, all sets $\Q$ of instances and all integer $k$ such that 
$$k \ge \mathsf{max} \{|B|, B \in LC(\Sigma^{\infty})\cup \mathcal Q \}$$  
$(\Sigma^{(k)},\Q)$ enjoys the flat chase property. 
\item[] {\emph{Proof Sketch.}} The property can be shown by adapting the proof of Fact 3. Indeed, a flat derivation by $\Sigma^{\infty}$ can be modified in two ways that preserves completeness:(1) one may replace the body of a tgd by a left-core of this tgd; and (2) one may select, for each tgd $B{\to}H$, the atoms from $H$ which will be used later in the derivation (as to enforce than $|H| \le k$). 
\end{itemize}
We can finally use Fact 5 to conclude the proof of Lemma~\ref{lem:5} and Theorem~\ref{theo:gflat}.
\end{proof} 

\section{Conclusion and Future Work}

This paper presented contributions along two axes:\\ 

{\bf More General Classes.}  Considering a complete resolution procedure extends the possible fields of applications of Datalog$^\pm$.  Resolution can be used, first, to generalize existing tractability criteria (such as stickiness) but it can also be used, in practice, without any syntactic assumption.  A technique of integration or simulation however proves often useful (and sometimes necessary) to handle egds. Datalog rewriting finally covers the three main paradigms of Datalog$^\pm$: stickiness, termination and guardedness. There are also alternative paradigms of tractability that were considered in \cite{BLMS09} which have not been discussed here (to simplify the discussion) but are yet to compare with the class of Datalog-rewritable dependencies. In particular, a natural question (left as future work) is the following: is Datalog rewriting always possible under  \emph{bounded-tree width}\cite{Cour90}?  \\    

{\bf More Efficient Algorithms.} This paper introduced several algorithms and heuristics  that can be of clear practical use. In particular,  the resolution can certainly (at least, in some contexts) be more efficient than the chase. The technique of Datalog Rewriting can also prove useful in practice. In particular, one may consider Magic Set or any other standard optimisation of Datalog to improve efficiency. Nonetheless, there is still a long road ahead because the algorithms of Datalog Rewriting proposed in Section 4 are (for the moment) very much non-optimal. An interesting and challenging question that remains is the following: how to compute \emph{in practice} a Datalog rewriting of reasonable size?      

\bibliographystyle{abbrv}

{\small

\bibliography{rew}

\begin{thebibliography}{10}

\bibitem{AbHV95}
S.~Abiteboul, R.~Hull, and V.~Vianu.
\newblock {\em Foundations of Databases}.
\newblock Addison-Wesley, 1995.

\bibitem{AV91}
S.~Abiteboul and V.~Vianu.
\newblock Datalog extensions for database queries and updates.
\newblock {\em J. Comput. Syst. Sci.}, 43(1):62--124, 1991.

\bibitem{ACDL*05}
A.~Acciarri, D.~Calvanese, G.~D. Giacomo, D.~Lembo, M.~Lenzerini, M.~Palmieri,
  and R.~Rosati.
\newblock {QuOnto}: Querying ontologies.
\newblock In {\em AAAI}, pages 1670--1671, 2005.

\bibitem{BLMS09}
J.-F. Baget, M.~Lecl{\`e}re, M.-L. Mugnier, and E.~Salvat.
\newblock Extending decidable cases for rules with existential variables.
\newblock In {\em IJCAI}, pages 677--682, 2009.

\bibitem{BaGO10}
V.~Barany, G.~Gottlob, and M.~Otto.
\newblock Querying the guarded fragment.
\newblock In {\em LICS}, 2010.

\bibitem{BeVa81}
C.~Beeri and M.~Y. Vardi.
\newblock The implication problem for data dependencies.
\newblock In {\em ICALP}, pages 73--85, 1981.

\bibitem{SemWeb}
F.~Bry, T.~Furche, B.~Marnette, C.~Ley, B.~Linse, and O.~Poppe.
\newblock {SPARQL}og: {SPARQL} with rules and quantification.
\newblock In {\em Semantic Web Information Management: A Model-Based
  Perspective}, pages 341 -- 369. Springer, 2010.

\bibitem{CaGK08}
A.~Cal\`{\i}, G.~Gottlob, and M.~Kifer.
\newblock Taming the infinite chase: Query answering under expressive
  relational constraints.
\newblock In {\em KR}, pages 70--80, 2008.

\bibitem{CaGL09}
A.~Cal\`{\i}, G.~Gottlob, and T.~Lukasiewicz.
\newblock A general {D}atalog-based framework for tractable query answering
  over ontologies.
\newblock In {\em PODS}, pages 77--86, 2009.

\bibitem{CGLMP10}
A.~Cal\`{\i}, G.~Gottlob, T.~Lukasiewicz, B.~Marnette, and A.~Pieris.
\newblock Datalog+/-: A family of logical knowledge representation and query
  languages for new applications.
\newblock In {\em LICS}, pages 228--242, 2010.

\bibitem{CGP10}
A.~Cal\`{\i}, G.~Gottlob, and A.~Pieris.
\newblock Query answering under non-guarded rules in datalog+/-.
\newblock In {\em RR}, pages 1--17, 2010.

\bibitem{CeGT90}
S.~Ceri, G.~Gottlob, and L.~Tanca.
\newblock {\em Logic Programming and Databases}.
\newblock Springer, 1990.

\bibitem{Cour90}
B.~Courcelle.
\newblock The monadic second-order logic of graphs. i. recognizable sets of
  finite graphs.
\newblock {\em Inf. Comput.}, 85(1):12--75, 1990.

\bibitem{ChaseRevisited}
A.~Deutsch, A.~Nash, and J.~B. Remmel.
\newblock The chase revisited.
\newblock In {\em PODS}, pages 149--158, 2008.

\bibitem{FKMP05}
R.~Fagin, P.~G. Kolaitis, R.~J. Miller, and L.~Popa.
\newblock Data exchange: Semantics and query answering.
\newblock {\em Theor. Comput. Sci.}, 336(1):89--124, 2005.

\bibitem{GuardProc}
E.~Gr{\"a}del.
\newblock Decision procedures for guarded logics.
\newblock In {\em CADE}, pages 31--51, 1999.

\bibitem{H02}
J.~Hladik.
\newblock Implementation and optimisation of a tableau algorithm for the
  guarded fragment.
\newblock In {\em TABLEAUX}, pages 145--159, 2002.

\bibitem{JoKl84}
D.~S. Johnson and A.~C. Klug.
\newblock Testing containment of conjunctive queries under functional and
  inclusion dependencies.
\newblock {\em J. Comput. Syst. Sci.}, 28(1):167--189, 1984.

\bibitem{KazNiv04}
Y.~Kazakov and H.~de~Nivelle.
\newblock A resolution decision procedure for the guarded fragment with
  transitive guards.
\newblock In {\em IJCAR}, volume 3097 of {\em Lecture Notes in Computer
  Science}, pages 122--136. Springer, 2004.

\bibitem{brunoPODS}
B.~Marnette.
\newblock Generalized schema-mappings: from termination to tractability.
\newblock In {\em PODS}, pages 13--22, 2009.

\bibitem{mythesis}
B.~Marnette.
\newblock {\em Tractable Schema Mappings Under Oblivious Termination}.
\newblock PhD thesis, Oxford University, 2010.

\bibitem{myVLDB}
B.~Marnette, G.~Mecca, and P.~Papotti.
\newblock Scalable data exchange with functional dependencies.
\newblock {\em PVLDB}, 3(1):105--116, 2010.

\bibitem{CompSkol}
A.~Nash, P.~A. Bernstein, and S.~Melnik.
\newblock Composition of mappings given by embedded dependencies.
\newblock In {\em PODS}, pages 172--183, 2005.

\bibitem{Nivelle03}
H.~D. Nivelle and M.~D. Rijke.
\newblock Deciding the guarded fragments by resolution.
\newblock {\em Journal of Symbolic Computation}, 35:21--58, 2003.

\bibitem{pumh08rewriting}
H.~P{\'e}rez-Urbina, B.~Motik, and I.~Horrocks.
\newblock {Rewriting Conjunctive Queries over Description Logic Knowledge
  Bases}.
\newblock In {\em Proc. of the Int.\ Workshop on Semantics in Data and
  Knowledge Bases (SDKB 2008)}, Nantes, France, March 2008. Springer.

\bibitem{Termin}
F.~Spezzano and S.~Greco.
\newblock Chase termination: A constraints rewriting approach.
\newblock {\em PVLDB}, 3(1):93--104, 2010.

\end{thebibliography}

}

\end{document}